\documentclass[11pt]{article}

\usepackage{graphicx,amssymb,amsmath}
\usepackage{amsfonts}
\usepackage{booktabs}
\usepackage{cite}
\usepackage{color}
\usepackage{enumerate}
\usepackage{float}

\usepackage[pdftex, plainpages = false, pdfpagelabels,
                 bookmarks=false,
                 bookmarksopen = true,
                 bookmarksnumbered = true,
                 breaklinks = true,
                 linktocpage,
                 pagebackref,
                 colorlinks = true,  
                 linkcolor = blue,
                 urlcolor  = cyan,
                 citecolor = red,
                 anchorcolor = green,
                 hyperindex = true,
                 hyperfigures
                 ]{hyperref}

\usepackage{mathrsfs}
\usepackage{subfig}
\usepackage{tabularx}
\usepackage{authblk}

\usepackage[in]{fullpage}
\usepackage[top=0.8in, bottom=0.9in, left=0.9in, right=0.9in]{geometry}

\def\calC{\mathcal{C}}

\def\calD{\mathcal{D}}

\newtheorem{observation}{Observation}

\newtheorem{lemma}{Lemma}
\newtheorem{theorem}{Theorem}
\newenvironment{proof}{\noindent {\textbf{Proof:}}\rm}{\hfill $\Box$ \rm\bigskip}

\title{Counting Unit Circular Arc Intersections\thanks{A preliminary version of this paper will appear in {\em Proceedings of 43rd International Symposium on Theoretical Aspects of Computer Science (STACS 2026)}.}}

\author{
Haitao Wang\thanks{Kahlert School of Computing,
University of Utah, Salt Lake City, UT 84112, USA. {\tt haitao.wang@utah.edu}}
}

\begin{document}

\pagestyle{plain}
\pagenumbering{arabic}
\setcounter{page}{1}
\date{}

\thispagestyle{empty}
\maketitle

\vspace{-0.35in}
\begin{abstract}
Given a set of $n$ circular arcs of the same radius in the plane, we consider the problem of computing the number of intersections among the arcs. The problem was studied before and the previously best algorithm solves the problem in $O(n^{4/3+\epsilon})$ time [Agarwal, Pellegrini, and Sharir, SIAM J. Comput., 1993],
for any constant $\epsilon>0$. No progress has been made on the problem for more than 30 years. We present a new algorithm of $O(n^{4/3}\log^{16/3}n)$ time and improve it to $O(n^{1+\epsilon}+K^{1/3}n^{2/3}(\frac{n^2}{n+K})^{\epsilon}\log^{16/3}n)$ time for small $K$,
where $K$ is the number of intersections of all arcs.
\end{abstract}



\section{Introduction}
\label{sec:intro}

Given a set of $n$ circular arcs of the same radius in the plane, we consider the problem of computing the number of intersections among the arcs; here we count the number of intersecting points (i.e., if two arcs intersect at two points, then they are counted twice).
Agarwal, Pellegrini, and Sharir~\cite{ref:AgarwalCo93} previously solved the
problem in $O(n^{4/3+\epsilon})$ time; throughout the
paper, we let $\epsilon$ denote an arbitrarily small positive constant.
On the negative side, Erickson's
results~\cite{ref:EricksonNe96} show that $\Omega(n^{4/3})$ is a lower bound for
the problem in a so-called partition algorithm model. Note that Erickson's
lower bound~\cite{ref:EricksonNe96} was originally proved for counting the intersections among a set of line segments in the plane. Since a line segment can be considered as an arc of a special circle of radius $\infty$, the lower bound is also applicable to the circular case.
No progress has been made on the problem for more than 30 years.
In this paper, we present a new algorithm of $O(n^{4/3}\log^{16/3}n)$ time and thus improve the previous result in~\cite{ref:AgarwalCo93}. Although the improvement looks minor, a factor of $O(n^{\epsilon})$, it is significant from the following perspective: it reduces the gap between the upper and lower bounds from
$\text{poly}(n)$ to $\text{poly}(\log n)$, where $\text{poly}(\cdot)$ is a polynomial function.
Further, for small $K$, where $K$ is the number of intersections of all arcs, we
improve the algorithm to
$O(n^{1+\epsilon}+K^{1/3}n^{2/3}(\frac{n^2}{n+K})^{\epsilon}\log^{16/3}n)$ time. Note that  this matches the above $O(n^{4/3}\log^{16/3}n)$ complexity in the worst case when $K=\Theta(n^2)$. If $K=O(n^{2-\delta})$ for any $\delta>0$, then the time complexity is $o(n^{4/3})$.

In addition, our results can be used to solve the following
{\em bichromatic problem} (with the same time complexity): Given a set of red circular arcs
and another set of blue circular arcs such that all arcs (of both colors) have the
same radius, compute the number of red-blue intersections.


\paragraph{Related work.}
If the circular arcs have different radii, Agarwal, Pellegrini, and Sharir~\cite{ref:AgarwalCo93} gave an $O(n^{3/2+\epsilon})$ time algorithm to compute the number of intersections of all arcs.

For computing the number of intersections among a set of $n$ circles of the same
radius, Katz and Sharir~\cite{ref:KatzAn97} gave an algorithm of $O(n^{4/3}\log
n)$ time and the randomized algorithm of Agarwal, Aronov, Sharir, and
Suri~\cite{ref:AgarwalSe93} can solve it in $O(n^{4/3}\log^{2/3}n)$ expected
time. Recently, Wang~\cite{ref:WangUn23} derived an $O(n^{4/3})$
time algorithm for it, matching the $\Omega(n^{4/3})$ lower bound~\cite{ref:EricksonNe96}.
If the circles have different radii, then the problem is
solvable in $O(n^{3/2+\epsilon})$ time~\cite{ref:AgarwalCo93}.
Recently, Chan, Cheng, and Zheng~\cite{ref:ChanSe24} extended the result and gave an $O(n^{3/2+\epsilon})$-time algorithm to compute the number of intersections among a set of $n$ algebraic arcs of constance description complexity in the plane.

The case where all arcs are line segments has been extensively studied.
First of all, Erickson's $\Omega(n^{4/3})$-time lower bound~\cite{ref:EricksonNe96}
still applies. On the positive side, Chazelle~\cite{ref:ChazelleRe86} gave an $O(n^{1.695})$ time
algorithm for the problem. Guibas, Overmars, and Sharir~\cite{ref:GuibasCo89} developed a
randomized algorithm of $O(n^{4/3+\epsilon})$ expected time.
Agarwal~\cite{ref:AgarwalPa902} solved the problem in $O(n^{4/3}\log^{1.78}n)$
time and Chazelle~\cite{ref:ChazelleCu93} further improved the algorithm to
$O(n^{4/3}\log^{1/3}n)$ time. For small $K$, where $K$ is the number of all
segment intersections, Pellegrini~\cite{ref:PellegriniOn97} gave an
algorithm of $O(n^{1+\epsilon}+K^{1/3}n^{2/3+\epsilon})$ time.
De Berg and Schwarzkopf~\cite{ref:deBergCu95} presented a randomized algorithm
of $O(n\log n\log K+K^{1/3}n^{2/3}\log^{1/3}n)$ expected time.
Recently, Chan and Zheng~\cite{ref:ChanHo23} made a breakthrough and solved the
problem in $O(n^{4/3})$ time, matching the $\Omega(n^{4/3})$ lower
bound~\cite{ref:EricksonNe96}.

With Chan and Zheng's new result~\cite{ref:ChanHo23} and our techniques, we can
show that the problem is solvable in
$O(n^{1+\epsilon}+K^{1/3}n^{2/3}(\frac{n^2}{n+K})^{\epsilon})$
time, or in $O(n\log n\log K+K^{1/3}n^{2/3})$ expected time
with the randomized techniques of de Berg and Schwarzkopf~\cite{ref:deBergCu95}.


One essential difference between the segment case and the circular arc case is
that there is at most one intersection between two segments while it is possible
to have two intersections between two arcs. This difference makes the technique
of Chan and Zheng's new result~\cite{ref:ChanHo23} not applicable to the arc
case. Indeed, this difference also makes the circular arc case much more challenging.

\medskip
In addition to a polynomial-time improvement over the 30-year old previous work~\cite{ref:AgarwalCo93} on a fundamental problem, our contribution also lies on many new geometric observations and techniques on the unit circular arcs. These new techniques may find applications in solving other related problems on unit circular arcs as well.

\paragraph{Outline.} The rest of the paper is organized as follows. We first
introduce some notation and concepts in Section~\ref{sec:pre}.
Section~\ref{sec:arc} presents our algorithm for computing the number of
intersections among circular arcs.
We conclude the paper in Section~\ref{sec:con} with remarks on the bichromatic problem as well as the segment case.

\section{Preliminaries}
\label{sec:pre}

In this section, we introduce some notation and concepts that will be used later.

For any two points $a$ and $b$ in the plane, we use $\overline{ab}$ to denote the segment with endpoints $a$ and $b$; we call $\overline{ab}$ the {\em defining segment} of $a$ (resp., $b$).
For any compact region $A$ in the plane, we use $\partial A$ to denote the boundary of $A$.


\paragraph{Cuttings.}
Let $H$ be a set of $n$ lines in the plane. For a compact region $A$ in the plane, we use $H_A$ to denote the subset of lines of $H$ that intersect the interior of $A$ (we also say that these lines {\em cross} $A$).
A {\em cutting} is a collection $\Xi$ of closed cells (each of which is a triangle, possibly unbounded) with disjoint interiors, which together cover the entire plane~\cite{ref:ChazelleCu93,ref:MatousekRa93}.
The {\em size} of $\Xi$ is the number of cells in $\Xi$.
For a parameter $r$ with $1\leq r\leq n$, a {\em $(1/r)$-cutting} for $H$ is a cutting $\Xi$ satisfying $|H_{\sigma}|\leq n/r$ for every cell $\sigma\in \Xi$.


We say that a cutting $\Xi'$ {\em $c$-refines} a cutting $\Xi$ if every cell of $\Xi'$ is contained in a single cell of $\Xi$ and every cell of $\Xi$ contains at most $c$ cells of $\Xi'$. Let $\Xi_0,\Xi_1,\ldots,\Xi_k$ be a sequence of cuttings for $H$ such that $\Xi_0$ is the entire plane, and every $\Xi_i$ is a $(1/\rho^i)$-cutting of size $O(\rho^{2i})$ which $c$-refines $\Xi_{i-1}$, for two constants $c$ and $\rho>1$. In order to make $\Xi_k$ a $(1/r)$-cutting, we set $k=\lceil\log_{\rho} r\rceil$.
The above sequence of cuttings is called a {\em hierarchical $(1/r)$-cutting} for $H$. If a cell $\sigma\in \Xi_{j-1}$ contains a cell $\sigma'\in \Xi_j$, we say that $\sigma$ is the {\em parent} of $\sigma'$ and $\sigma'$ is a {\em child} of $\sigma$. As such, one could view $\Xi$ as a tree structure with $\Xi_0$ as the root.

For any $1\leq r\leq n$, a hierarchical $(1/r)$-cutting of size $O(r^2)$ for $H$ (together with the sets $H_{\sigma}$ for every cell $\sigma$ of $\Xi_i$ for all $i=0,1,\ldots,k$) can be computed in $O(nr)$ time~\cite{ref:ChazelleCu93}. Note that this implies that the total size of $H_{\sigma}$ for all cells $\sigma\in \Xi_i$, $0\leq i\leq k$, is $O(nr)$.

As indicated by Agarwal and Sharir~\cite{ref:AgarwalPs05},
similar results on cuttings for other more general curves in the plane (e.g., circles or circular arcs of different radii, pseudo-lines, line segments) can also be obtained by extending Chazelle's algorithm~\cite{ref:ChazelleCu93}; Wang~\cite{ref:WangUn23} provided algorithm details for that. For our purpose, if $H$ is a set of $n$ circular arcs (not necessarily with the same radius) in the plane, then we can also define cuttings for $H$ in the same way as above.  A difference is that a cell $\sigma$ of a cutting becomes a {\em pseudo-trapezoid} (more specifically, each pseudo-trapezoid in general has four sides, with the left/right side as a vertical line segment and the top/bottom side as part of an input circular arc; e.g., see Fig.~\ref{fig:pseduotrapezoid}). In particular, the following result, which will later be used in our algorithm, has been obtained in~\cite{ref:WangUn23}.

\begin{theorem}\label{theo:cutting}{\em\cite{ref:WangUn23}}
Suppose $H$ is a set of $n$ circular arcs (or line segments) in the plane and
$K$ is the number of intersections of the arcs of $H$.
For any $r\leq n$, a hierarchical $(1/r)$-cutting of size $O(r^2)$ for $H$ (together with the sets $H_{\sigma}$ for every cell $\sigma$ of $\Xi_i$ for all $0\leq i\leq k$) can be computed in $O(nr)$ time; more specifically, the size of the cutting is bounded by $O(r^{1+\epsilon}+K\cdot r^2/n^2)$ and the running time of the algorithm is bounded by $O(nr^{\epsilon}+K\cdot r/n)$, for any small $\epsilon>0$.
\end{theorem}

\begin{figure}[t]
\begin{minipage}[t]{\textwidth}
\begin{center}
\includegraphics[height=1.0in]{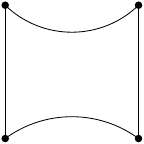}
\caption{\footnotesize Illustrating a pseudo-trapezoid. The top (resp., bottom) side is a circular arc that is part of an input arc. The left (resp., right) side a vertical segment.}
\label{fig:pseduotrapezoid}
\end{center}
\end{minipage}
\vspace{-0.15in}
\end{figure}

Cuttings are one of our main tools. Although cuttings is a standard technique, the novelty of our approach lies in how to combine cuttings with the newly discovered geometric observations and techniques on unit circular arcs.

\section{The algorithm}
\label{sec:arc}

In this section, we present an $O(n^{4/3}\log^{16/3} n)$-time algorithm for computing the number of intersections of $n$ circular arcs of the same radius in the plane. At the very end of this section, we further reduce the time of the algorithm to $O(K^{1/3}n^{2/3}(\frac{n^2}{n+K})^{\epsilon}\log^{16/3}n)$, where $K$ is the number of intersections of all arcs.

Let $S$ be a set of $n$ circular arcs of the same radius in the plane. Without loss of generality, we assume that the radius is $1$, and each arc is thus a {\em unit circular arc}. Our goal is to compute the number of intersections of all arcs of $S$. In the following, unless otherwise stated, a circular arc refers to a unit circular arc.

We call a disk a {\em unit disk} if its radius is $1$; the boundary of a unit disk is a {\em unit circle}.
For each arc $s$, the circle that contains it is called its {\em underlying circle} (the disk bounded by the circle is called the {\em underlying disk}), and the center of the circle is also called the {\em center} of $s$. We use $\alpha(s)$ and $D(s)$ to denote the underlying circle and disk of $s$, respectively.
Let $P$ denote the set of the centers of the arcs of $S$.
For any region $A$ in the plane, let $P(A)$ denote the subset of points of $P$ in $A$, i.e., $P(A)=P\cap A$.

\begin{figure}[t]
\begin{minipage}[t]{\textwidth}
\begin{center}
\includegraphics[height=1.6in]{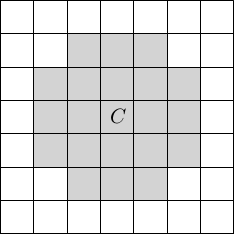}
\caption{\footnotesize The set $N(C)$ is a subset of the grey cells. Because the side-length of each cell is $1/\sqrt{2}$, the length between any point in $C$ and any point not in a grey cell must be larger than $1$.}
\label{fig:grid}
\end{center}
\end{minipage}
\vspace{-0.15in}
\end{figure}

With respect to the point set $P$, Wang gave an algorithm in~\cite{ref:WangUn23} that can compute a set $\calC$ of $O(n)$ axis-parallel square cells in the plane whose interiors are pairwise disjoint and whose union covers all the points of $P$, with the following properties: (1) Each cell has side length $1/\sqrt{2}$. (2) Every two cells are separated by an axis-parallel line. (3) For a unit disk $D_p$ with center $p$, if $p$ is not in any cell of $\calC$, then $D_p\cap P=\emptyset$.
(4) Each cell $C$ of $\calC$ is associated with a subset $N(C)$ of $O(1)$ cells of $\calC$, such that for any disk $D$ with center in $C$, every point of $P\cap D$ is in one of the cells of $N(C)$; e.g., see Fig.~\ref{fig:grid}.
(5) Each cell $C'$ of $\calC$ is in $N(C)$ for a constant number of cells $C\in \calC$.
In addition to $\calC$, Wang's algorithm~\cite{ref:WangUn23} also computes the subsets $P(C)$ and $N(C)$ for all cells $C\in \calC$, in a total of $O(n\log n)$ time and $O(n)$ space.
Using the properties of $\calC$ described above, we can obtain the following lemma.

\begin{lemma}\label{lem:60}
\begin{enumerate}
  \item Eevery arc of $S$ is covered by the union of all cells of $\calC$.
  \item If an arc $s\in S$ intersects $C$, then the center of $s$ must be in one of the cells of $N(C)$.
  \item Each arc of $S$ only intersects a constant number of cells of $\calC$.
  \item $\sum_{C\in \calC}\sum_{C'\in N(C)}|P(C')|=O(n)$.
\end{enumerate}
\end{lemma}
\begin{proof}
To prove the first lemma statement, it suffices to show that for any point $p$ of an arc of $S$, $p$ must be in a cell $C$ of $\calC$. Assume to the contrary this is not the case for a point $p$ of an arc $s\in S$. Then, by property (3) of $\calC$, $D_p\cap P=\emptyset$. However, this is not true since $D_p$ contains the center of $s$, which is in $P$. We thus obtain contradiction.


For the second lemma statement, consider a point $p$ in a cell $C$ and $p$ is also on an arc $s\in S$. Let $q$ be the center of $s$. Since $D_p$ contains $q$, by the property (4) of $\calC$, $q$ must be in one of the cells of $N(C)$.

For the third lemma statement, consider an arc $s$ whose center is in a cell $C'$. If $s$ intersects a cell $C$, then by the second lemma statement, $C'$ must be in $N(C)$. According to the property (5) of $\calC$, each cell $C'$ of $\calC$ is in $N(C)$ for a constant number of cells $C\in \calC$. This implies that $s$ only intersects a constant number of cells of $\calC$.

The fourth lemma statement simply follows the property (5) of $\calC$.
\end{proof}

For each cell $C\in \calC$, we use $S(C)$ to denote the set of arcs of $S$ whose centers are in $P(C)$.

With Lemma~\ref{lem:60}, we will compute the number of intersections of the arcs of $S$ inside each cell $C$ of $\calC$ and then add these numbers together. In what follows, we focus on one such cell $C$. Our goal is to compute the number of intersections of the arcs of $S$ inside $C$. By Lemma~\ref{lem:60}, we only need to consider the arcs in $S(C')$ for all $C'\in N(C)$. To this end, for each pair of cells $(C_1,C_2)$ of $N(C)$, including the case where $C_1=C_2$, we will compute the number of intersections inside $C$ between the arcs of $S(C_1)$ and the arcs of $S(C_2)$; the sum of these numbers is the number of intersections inside $C$ among all arcs of $S$. In the following, we discuss our algorithm for such a pair $(C_1,C_2)$.

We assume that $C_1$ and $C_2$ are two different cells since the other case can be solved by the same techniques. Let $n_j=|S(C_j)|$, $j=1,2$. Since we are interested in the intersections inside $C$, for each arc $s\in S(C_j)$, we only keep its portions inside $C$. As the radius of $s$ is $1$ and $C$ is an axis-parallel square of side length $1/\sqrt{2}$, $s$ has at most two (maximal) sub-arcs in $C$. Let $S_j$ denote the set of these sub-arcs of $S(C_j)$. Thus, $|S_j|\leq 2n_j$. Our goal is to compute the number of intersections between the arcs of $S_1$ and the arcs of $S_2$.

\subsection{Counting intersections between $S_1$ and $S_2$: Algorithm overview}
\label{sec:counttwocells}

To simplify the notation, we temporarily let $n=|S_1|+|S_2|$ and $S=S_1\cup S_2$. By definition, each arc of $S$ is inside $C$.
For each arc $s\in S$, we extend both endpoints of $s$ along the underlying circle of $s$ until $\partial C$, and we refer to the new arc as the {\em extending arc} of $s$. Let $S'$ denote the set of all extending arcs of $S$. Note that every two arcs of $S'$ intersect at most twice.

We will start with computing a cutting of $S'$ and then solve a {\em pseudo-trapezoid-restricted subproblem} in each cell of the cutting. Below we first describe our algorithm for solving the restricted subproblem. Consider a pseudo-trapezoid $\tau$ of the cutting, which is in $C$.  Let $S(\tau)$ be the set of arcs of $S$ intersecting $\tau$ and let $n_{\tau}=|S(\tau)|$. Let $S_{1}(\tau)=S_1\cap S(\tau)$ and $S_{2}(\tau)=S_2\cap S(\tau)$.
The goal of the restricted subproblem is to compute the number of intersections between the arcs of $S_1(\tau)$ and $S_2(\tau)$ inside $\tau$. If we color the arcs of $S_1(\tau)$ red and color those of $S_2(\tau)$ blue, then we essentially have a bichromatic arc intersection problem: compute the number of intersections between red arcs and blue arcs.


We say that an arc $s\in S(\tau)$ is a {\em short arc} if it has at least one endpoint in the interior of $\tau$.
Otherwise, $s$ is a {\em long arc}.
Because the radius of $s$ is $1$, $\tau$ is a pseudo-trapezoid inside $C$ (and the radii of the circular arcs on the boundary of $\tau$ are also $1$), and $C$ a square cell of side-length $1/\sqrt{2}$, the intersection of $\tau$ and $s$ consists of at most two (maximal) sub-arcs.
Since we are only interested in intersections inside $\tau$, for each arc $s\in S(\tau)$, we trim it and only keep its portion (possibly two sub-arcs) inside $\tau$. If $s$ has two trimmed sub-arcs in $\tau$, then both sub-arcs are treated as independent arcs, which are considered as long (resp., short) arcs if $s$ is a long (resp., short) arc. A useful property of long arcs is that each long arc has both endpoints on $\partial \tau$.
To simplify the notation, we still use $S(\tau)$ to denote the set of all trimmed arcs in $\tau$ and use $n_{\tau}$ to denote the size of $S(\tau)$. Let $m_{\tau}$ denote the number of short arcs of $S(\tau)$.

The arc intersections that we are interested in can now be classified into four types: (1) {\em long-red and long-blue intersections} (i.e., intersections between long red arcs and long blue arcs), (2) {\em short-red and short-blue intersections}, (3) {\em long-red and short-blue intersections}, and (4) {\em short-red and long-blue intersections}.


For counting type (2) intersections, we use the algorithm of Agarwal, Pellegrini, and Sharir~\cite{ref:AgarwalCo93}, which runs in $O(m_{\tau}^{4/3+\epsilon})$ time. We will give an $O(n_{\tau}\log^2 m_{\tau}+m_{\tau}\sqrt{n_{\tau}}\log m_{\tau})$-time algorithm for counting type (3) intersections in Section~\ref{sec:type3}, and type (4) intersections can be counted in a similar way. Counting type (1) intersections is discussed in Section~\ref{sec:type1}.


\subsection{Counting type (3) intersections}
\label{sec:type3}

In this section, we present an algorithm that can compute the number of type (3) intersections in $O(n_{\tau}\log^2 m_{\tau}+m_{\tau}\sqrt{n_{\tau}}\log m_{\tau})$ time.



We further partition type (3) intersections into two sub-types: (3.1)
a pair of a long red arc and a short blue arc that intersect only once; (3.2) a pair of a long red arc and a short blue arc that intersect twice. We first discuss how to compute the number of type (3.1) intersections.

\subsubsection{Counting type (3.1) intersections}

We first present an algorithm that can compute the number of type (3.1) intersections in $O(m_{\tau}^2/\log m_{\tau}+n_{\tau}\log^2 m_{\tau})$ time.

For any long arc $s\in S(\tau)$, the intersection of its underlying circle $\alpha(s)$ and $\tau$ consists of at most two arcs (one of which is $s$). If the intersection is only one arc, which is $s$, then we call $s$ a {\em full arc}; otherwise it is a {\em partial arc}.

We further partition type (3.1) intersections into two sub-types: (3.1.1) a pair of a long red full arc and a short blue arc that intersect only once; (3.1.2) a pair of a long red partial arc and a short blue arc that intersect only once.
In what follows, we first describe an algorithm to compute the number of type (3.1.1) intersections. The algorithm for type (3.1.2) is more complicated.

\paragraph{Counting type (3.1.1) intersections.}
Consider a short blue arc $s_b$. Let $D_1$
and $D_2$ be the two unit disks centered at the two endpoints of
$s_b$, respectively. We call the region of $D_1$ that is not in $D_2$
a {\em lune}; similarly, region of $D_2$ that is not in $D_1$ is also
a {\em lune}. We use $lune(s_b)$ refer to the union of the two lunes (for simplicity, we consider $lune(s_b)$ as a single lune defined by $s_b$), i.e., $lune(s_b)=D_1\cup D_2-D_1\cap D_2$, the symmetric difference of $D_1$ and $D_2$ (e.g., the grey area in Fig.~\ref{fig:lune}). By definition, a point is in $lune(s_b)$ if and only if the unit disk centered at the point contains exactly one endpoint of $s_b$, implying that the unit circle centered at the point intersects $s_b$ exactly once.

The following observation is critical to our algorithm for counting type (3.1.1) intersections.

\begin{figure}[t]
\begin{minipage}[t]{\textwidth}
\begin{center}
\includegraphics[height=2.2in]{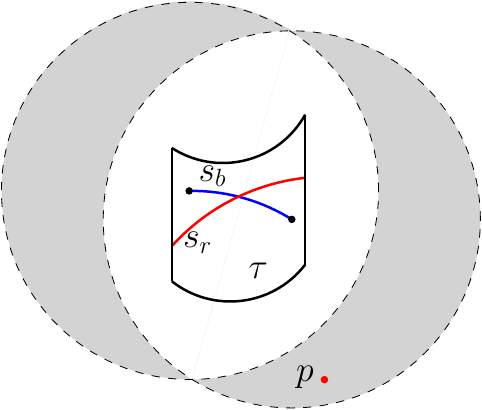}
\caption{\footnotesize Illustrating Observation~\ref{obser:20}: $p$ is the center of $s_r$. The two dashed circles are unit circles centered at the two endpoints of $s_b$, respectively. The grey area is $lune(s_b)$.}
\label{fig:lune}
\end{center}
\end{minipage}
\vspace{-0.15in}
\end{figure}

\begin{observation}\label{obser:20}
For a short blue arc $s_b$ and a long red full arc $s_r$, $s_r$ intersects $s_b$ exactly once if and only if the center of $s_r$ lies in $lune(s_b)$; e.g., see Fig.~\ref{fig:lune}.
\end{observation}
\begin{proof}
Let $p$ be the center of $s_r$ and $D$ be the underlying disk of $s_r$.

If $s_r$ intersects $s_b$ exactly once, then one endpoint of $s_b$ must be inside $D$ while the other is outside $D$. Therefore, $p$ must be in $lune(s_b)$. Below we prove the other direction of the observation.

If $p$ is in $lune(s_b)$, then by the definition of $lune(s_b)$, one endpoint of $s_b$ is in $D$ while the other is not. Hence, $s_b$ must cross $\partial D$ at exactly one point, say, $q$, which is the only intersection between $s_b$ and $\partial D$. As $s_b$ is in the cell $\tau$, $q\in \tau$. Since $q\in \partial D$, we obtain that $q\in \partial D\cap \tau$. Further, since $s_r$ is a long full arc, we have $s_r=\partial D\cap \tau$. Hence, $q\in s_r$. This implies that $q\in s_r\cap s_b$. Since $s_r\in \partial D$ and $s_b$ and $\partial D$ has exactly one intersection, we conclude that $s_b$ and $s_r$ has exactly one intersection.
\end{proof}

Based on the above observation, we have the following lemma.

\begin{lemma}\label{lem:70}
The number of type (3.1.1) intersections can be computed in $O(m_{\tau}^2/\log m_{\tau}+n_{\tau}\log m_{\tau})$ time.
\end{lemma}
\begin{proof}
For notational convenience, we let $n=n_{\tau}$ and $m=m_{\tau}$ in this proof.
Recall that centers of red arcs are in the square cell $C_r\in \calC$. For each short blue arc $s_b$, since the radius of $s_b$ is $1$ and $C_r$ is a square cell of the side-length $1/\sqrt{2}$, the boundary of the intersection of $C_r$ and $lune(s_b)$ has at most four circular arcs (which are on the unit circles centered at the two endpoints of $s_b$, respectively); let $H$ be the set of these circular arcs defined by all short blue arcs. Hence, $|H|=O(m)$.
In light of Observation~\ref{obser:20}, it suffices to compute the number of lunes $lune(s_b)$ containing the center of each long red arc. Let $L$ denote the set of the lunes $lune(s_b)$ of all short blue arcs $s_b$.

We build a hierarchical $(1/r)$-cutting $\Xi_0,\Xi_1,\ldots,\Xi_k$ on the arcs
of $H$, with $r=m/\log m$.
This can be done in $O(mr)$
time by Theorem~\ref{theo:cutting}. For each cell $\sigma\in \Xi_i$,
$1\leq i\leq k$, we define $m_{\sigma}$ as the number of lunes of $L$ that
contain $\sigma$ but do not contain the parent cell of $\sigma$ in $\Xi_{i-1}$
(i.e., the boundary of the lune crosses the parent cell). We can compute
$m_{\sigma}$ for all cells $\sigma\in \Xi_i$, for all $i=1,2,\ldots,k$, as
follows.

Recall that the algorithm of Theorem~\ref{theo:cutting} for computing the
hierarchical cutting also produces the sets $H_{\sigma}$ for all cells $\sigma\in
\Xi_i$, $0\leq i\leq k$, where $H_{\sigma}$ is the subset of arcs of $H$
crossing $\sigma$, with the same time complexity $O(mr)$. As discussed in
Section~\ref{sec:pre}, the total size of these sets for all cells is $O(mr)$.


Next, we show that the values $m_{\sigma}$ for all cells $\sigma\in \Xi_i$, $1\leq i\leq k$ can be computed in $O(mr)$ time using the sets $H_{\sigma}$. Define $L(\sigma)$ as the subset of lunes of $L$ that contain $\sigma$ but does not contain the parent cell of $\sigma$. Hence, $m_{\sigma}=|L(\sigma)|$. Observe that one of the bounding arcs of each lune of $L(\sigma)$ must cross $\sigma'$, where $\sigma'$ is the parent of $\sigma$ in $\Xi_{i-1}$. This implies that we can compute $L(\sigma)$ and thus $m_{\sigma}$ using $H_{\sigma'}$ as follows. For each arc $s\in H_{\sigma'}$, we check whether the lune of $L$ bounded by $s$ contains $\sigma$; if yes, we add the lune to $L(\sigma)$. Since each cell of the hierarchical cutting has $O(1)$ child cells, the above algorithm for computing $m_{\sigma}$ for all cells $\sigma\in \Xi_i$, $1\leq i\leq k$, runs in time linear in the total size of $H_{\sigma}$ for all cells $\sigma\in \Xi_i$, $0\leq i\leq k$, which is $O(mr)$ as explained above. As such, after the hierarchical cutting is computed, $m_{\sigma}$ for all cells $\sigma\in \Xi_i$, $1\leq i\leq k$, can be computed in $O(mr)$ time.

For each long red arc $s$, we locate the cell $\sigma_i$ of $\Xi_i$ containing $q$ for each $1\leq i\leq k$, where $q$ is the center of $s$; we add all these values $m_{\sigma_i}$ to a total count $m_s$ (initially, $m_s=0$). Finally, for the cell $\sigma\in \Xi_k$, for each arc $s'\in H_{\sigma}$, we check whether the lune of $L$ bounded by $s'$ contains $q$; if yes, we increases $m_s$ by one. It is not difficult to see the value $m_s$ thus obtained is equal to the number of lunes of $L$ containing $q$. Since $|H_{\sigma}|=O(m/r)=O(\log m)$ for any $\sigma\in \Xi_k$ and $k=O(\log r)$, the time of the above algorithm for computing $m_s$ is $O(\log m)$. We apply the above algorithm for each long red arc, after which the number of type (3.1.1) intersections is computed. As there are at most $n$ long red arcs, the total time of the algorithm is $O(mr+n\log m)$, which is $O(m^2/\log m+n\log m)$ as $r=m/\log m$.
\end{proof}

\paragraph{Counting type (3.1.2) intersections.}
We now discuss the type (3.1.2) intersections. Consider a long red partial arc $s_r$. Recall that the center of $s_r$ is in the square cell $C_r$. According to the property (2) of $\calC$, the two cells $C_r$ and $C$ are separated by an axis-parallel line. Depending on whether the line is horizontal or vertical, without loss of generality, we assume that $C_r$ is either below $C$ or to the right of $C$.
Recall that the pseudo-trapezoid $\tau$ is in $C$, and typically the upper (resp., lower) side of $\tau$ is a unit circular arc
while both the left and right sides of $\tau$ are vertical segments.
Since $s_r$ is a partial arc, the underlying circle $\alpha(s_r)$ intersects $C$ at another arc, denoted by $s_r'$; we call $s_r'$ the {\em coupled arc} of $s_r$ (and $s_r$ is also the coupled arc of $s_r'$).

Consider a type (3.1.2) intersection between a long red partial arc $s_r$ and a short blue arc $s_b$.
We further partition type (3.1.2) intersections into two types: (3.1.2.1) $s_b$ does not intersect the coupled arc $s_r'$ of $s_r$; (3.1.2.2) $s_b$ intersects $s_r'$ (since $s_r$ and $s_r'$ are on the same circle and $s_b$ intersects $s_r$, $s_b$ intersects $s_r'$ exactly once in this case). We will present algorithms to count type (3.1.2.1) and type (3.1.2.2) intersections separately.

For each short-blue arc $s_b$, we use $lune'(s_b)$ to denote the region of the plane outside the union of the two unit disks centered at the two endpoints of $s_b$.

\begin{figure}[t]
\begin{minipage}[t]{\textwidth}
\begin{center}
\includegraphics[height=2.0in]{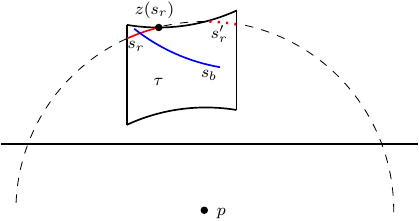}
\caption{\footnotesize Illustrating $s_r$, $s_r'$, and $z(s_r)$ for the case where $C_r$ is below $C$. The point $p$ is the center of $s_r$ and the dashed curve is the upper semicircle of $\alpha(s_r)$. In this figure, $s_r$ and $s_b$ form a type (3.1.2.1) intersection (i.e., $s_r$ and $s_b$ intersect exactly once and $s_r'$ does not intersect $s_b$).}
\label{fig:partial10}
\end{center}
\end{minipage}
\vspace{-0.15in}
\end{figure}

\begin{figure}[t]
\begin{minipage}[t]{\textwidth}
\begin{center}
\includegraphics[height=2.0in]{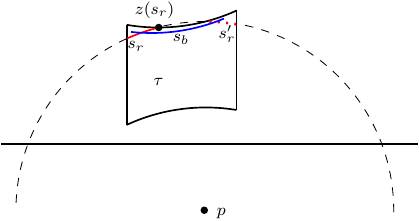}
\caption{\footnotesize Illustrating $s_r$, $s_r'$, and $z(s_r)$ for the case where $C_r$ is below $C$. The point $p$ is the center of $s_r$ and the dashed curve is the upper semicircle of $\alpha(s_r)$. In this figure, $s_r$ and $s_b$ form a type (3.1.2.2) intersection (i.e., $s_b$ and $s_r$ intersect exactly once and $s_b$ intersects $s_r'$).}
\label{fig:partial15}
\end{center}
\end{minipage}
\vspace{-0.15in}
\end{figure}

Recall that $C_r$ is either below $C$ or to the right of $C$. In the following, we give observations for both cases, which lead to algorithms for counting type (3.1.2.1) and type (3.1.2.2) intersections.

Suppose $C_r$ is below $C$. Then, $s_r$ and $s_r'$ are both on the upper semicircle of $\alpha(s_r)$ (e.g., see Fig.~\ref{fig:partial10}). Thus, both $s_r$ and $s_r'$ are $x$-monotone, and $s_r$ is either to the left or to the right of $s_r'$. Let $z(s_r)$ denote the endpoint of $s_r$ closer to $s_r'$, e.g., $z(s_r)$ is the right (resp., left) endpoint of $s_r$ if $s_r$ is to the left (resp., right) of $s_r'$.
We have the following observation for type (3.1.2.1) and type (3.1.2.2) intersections.

\begin{observation}\label{obser:30}
Let $s_r$ be a long red partial arc and $s_b$ a short blue arc.
Suppose $C_r$ is below $C$.
\begin{enumerate}
\item
$s_r$ and $s_b$ form a type (3.1.2.1) intersection (i.e., $s_r$ and $s_b$ intersect exactly once and $s_r'$ does not intersect $s_b$) if and only if the following two conditions hold: (1) the center of $s_r$ lies in $lune(s_b)$; (2) the endpoint of $s_b$ not in the disk $D(s_r)$ is to the left (resp., right) of $z(s_r)$ if $s_r$ is to the left (resp., right) of $s_r'$. See Fig.~\ref{fig:partial10}.

\item
$s_r$ and $s_b$ form a type (3.1.2.2) intersection (i.e., $s_b$ and $s_r$ intersect exactly once and $s_b$ intersects $s_r'$) if and only if the following two conditions hold:
(1) the center of $s_r$ lies in $lune'(s_b)$; (2) one endpoint of $s_b$ is to the left of $z(s_r)$ while the other endpoint is to the right of $z(s_r)$. See Fig.~\ref{fig:partial15}.
\end{enumerate}
\end{observation}
\begin{proof}
We assume that $s_r$ is to the left of $s_r'$ since the argument for the other case is analogous.
Thus $z(s_r)$ is the right endpoint of $s_r$ (e.g., see Fig.~\ref{fig:partial10}). Let $c(s_r)$ be the center of $s_r$.

Since $C_r$ is below $C$, and $s_r$ and the arcs of the upper and lower sides of $\tau$ all have radii equal to $1$, one can verify that $z(s_r)$ must be on the upper side of $\tau$ (this also holds if the upper and/or lower sides of $\tau$ are line segments).
As $s_r$ is a long arc, $s_r$ has both endpoints on $\partial \tau$ and thus separates $\tau$ into two regions, one locally inside $D(s_r)$ and the other locally outside $D(s_r)$; we use $\tau(s_r)$ to refer to the one locally outside $D(s_r)$.
Similarly, $s'_r$ separates $\tau$ into two regions and we use $\tau(s_r')$ to refer to the one locally outside $D(s_r')$. Let $\tau(s_r,s_r')$ to denote the region of $\tau$ not in $\tau(s_r)\cup \tau(s_r')$. Observe that $\tau(s_r,s_r')$ is exactly the intersection of $C$ and the underlying disk $D(s_r)$ of $s_r$. Since the left and right sides of the pseudo-trapezoid $\tau$ are vertical segments, both the upper and lower sides of $\tau$ must be $x$-monotone. Hence, $z(s_r)$ must be the rightmost point of $\tau(s_r)$. Similarly, the left endpoint of $s_r'$ is the leftmost point of $\tau(s_r')$. As $s_r$ is to the left of $s_r'$, all points of $\tau(s_r')$ are to the right of $z(s_r)$.

\begin{figure}[t]
\begin{minipage}[t]{\textwidth}
\begin{center}
\includegraphics[height=1.4in]{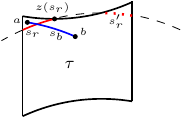}
\caption{\footnotesize Illustrating the case where $s_r$ and $s_b$ intersect exactly once while $s_r'$ does not intersect $s_b$.}
\label{fig:partial30}
\end{center}
\end{minipage}
\vspace{-0.15in}
\end{figure}

We first prove the first statement of the observation about type (3.1.2.1) intersections.

\begin{itemize}
  \item If $s_r$ and $s_b$ intersect exactly once while $s_r'$ does not intersect $s_b$, then one endpoint of $s_b$ (denoted by $a$) must be in $\tau(s_r)$ and the other one (denoted by $b$) must be in $\tau(s_r,s_r')$; see Fig.~\ref{fig:partial30}. This implies that $b$ is in $D(s_r)$ while $a$ is outside $D(s_r)$. Hence, the center of $s_r$ must be in $lune(s_b)$. Further, since $a\in \tau(s_r)$ and $z(s_r)$ is the rightmost point of $\tau(s_r)$, $a$ must be to the left of $z(s_r)$. This proves the ``only if'' direction of the statement.
  \item We now prove the ``if'' direction of the statement. Suppose the center of $s_r$ is in $lune(s_b)$ and $a$ is to the left of $z(s_r)$, where $a$ is the endpoint of $s_b$ not in $D(s_r)$. Let $b$ be the other endpoint of $s_b$ than $a$. Then, $b$ is in $D(s_r)$. Since all points of $\tau(s_r')$ are to the right of $z(s_r)$ and $a$ is to the left of $z(s_r)$, $a$ cannot be in $\tau(s_r')$. Since $a$ is not in $D(s_r)$, we obtain that $a\not\in D(s_r)\cap C=\tau(s_r,s_r')$. Therefore, $a$ must be in $\tau(s_r)$. Since $b\in D(s_r)$, $b$ is in $D(s_r)\cap C=\tau(s_r,s_r')$. As $a$ and $b$ are two endpoints of $s_b$, we conclude that $s_b$ intersects $s_r$ exactly once.
\end{itemize}

We next prove the second statement of the observation about type (3.1.2.2) intersections.


\begin{itemize}
\item
If $s_r$ and $s_b$ intersect exactly once and $s_b$ also intersects $s_r'$, then as discussed before $s_b$ intersects $s_r'$ exactly once; see Fig.~\ref{fig:partial15}. This implies that $s_b$ has an endpoint in $\tau(s_r)$ and an endpoint in $\tau(s_r')$. Because both $\tau(s_r)$  and $\tau(s_r')$ are outside $D(s_r)$, both endpoints of $s_b$ are outside $D(s_r)$, and thus the center of $s_r$ must be in $lune'(s_b)$. Let $a$ be the endpoint of $s_b$ in $\tau(s_r)$ and $b$ be the other endpoint, which is in $\tau(s_r')$. As discussed above, $a$ must be to the left of $z(s_r)$ and $b$ must be to the right of $z(s_r)$. This proves the ``only if'' direction of the statement.

\item
We now prove the ``if'' direction. Suppose the center of $s_r$ lies in $lune'(s_b)$ and one endpoint of $s_b$ is to the left of $z(s_r)$ while the other one is to the right of $z(s_r)$. Then, both endpoints of $s_b$ are outside $D(s_r)$, implying that the two endpoints are in $\tau(s_r)\cup \tau(s_r')$. Further, since one endpoint of $s_b$ is to the left of $z(s_r)$ while the other one is right of $z(s_r)$ and also because all points of $\tau(s_r)$ are to the left of $z(s_r)$ and all points of $\tau(s_r')$ are to the right of $z(s_r)$, we can obtain that one endpoint of $s_b$ is in $\tau(s_r)$ while the other one is in $\tau'(s_r)$. This implies that $s_b$ intersects each of $s_r$ and $s_r'$ exactly once.
\end{itemize}

This proves the observation.
\end{proof}

\begin{figure}[t]
\begin{minipage}[t]{\textwidth}
\begin{center}
\includegraphics[height=3.6in]{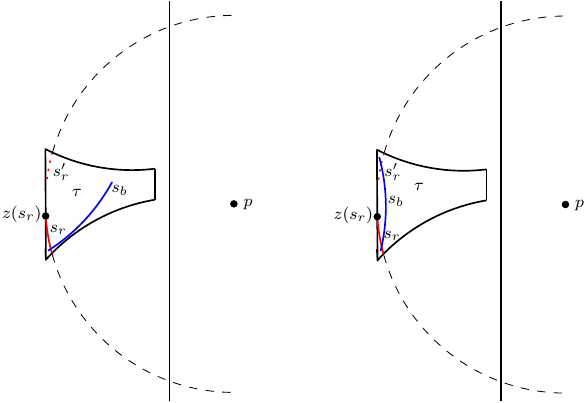}
\caption{\footnotesize Illustrating $s_r$, $s_r'$, and $z(s_r)$ where $C_r$ is right of $C$. The point $p$ is the center of $s_r$ and the dashed curve is the left semicircle of $\alpha(s_r)$. Left: $s_r$ and $s_b$ form a type (3.1.2.1) intersection (i.e., $s_r$ and $s_b$ intersect exactly once and $s_r'$ does not intersect $s_b$). Right: $s_r$ and $s_b$ form a type (3.1.2.2) intersection (i.e., $s_b$ and $s_r$ intersect exactly once and $s_b$ intersects $s_r'$).}
\label{fig:partial27}
\end{center}
\end{minipage}
\vspace{-0.15in}
\end{figure}


We now consider the other case where $C_r$ is to the right of $C$. In this case, $s_r$ and $s_r'$ are both in the left semicircle of $\alpha(s_r)$; e.g., see Fig.~\ref{fig:partial27}. Thus, both $s_r$ and $s_r'$ are $y$-monotone, and $s_r$ is either above or below $s_r'$. Let $z(s_r)$ denote the endpoint of $s_r$ closer to $s_r'$, e.g., $z(s_r)$ is upper (resp., lower) endpoint of $s_r$ if $s_r$ is below (resp., above) $s_r'$. We have the following observation, which is similar in spirit to Observation~\ref{obser:30}.

\begin{observation}\label{obser:40}
Let $s_r$ be a long-red partial arc and $s_b$ a short-blue arc.
Suppose $C_r$ is to the right of $C$.
\begin{enumerate}
  \item
   $s_r$ and $s_b$ form a type (3.1.2.1) intersection if and only if the following two conditions hold: (1) the center of $s_r$ lies in $lune(s_b)$; (2) the endpoint of $s_b$ not in the disk $D(s_r)$ is below (resp., above) $z(s_r)$ if $s_r$ is below (resp., above) $s_r'$.
   See Fig.~\ref{fig:partial27} (left).
   \item
   $s_r$ and $s_b$ form a type (3.1.2.2) intersection if and only if the following two conditions hold: (1) the center of $s_r$ lies in $lune'(s_b)$; (2) one endpoint of $s_b$ is above $z(s_r)$ while the other endpoint is below $z(s_r)$.
   See Fig.~\ref{fig:partial27} (right).
\end{enumerate}
\end{observation}
\begin{proof}
The proof is similar to Observation~\ref{obser:30} and we only point out the differences here. Without loss of generality, we assume that $s_r$ is below $s_r'$, and thus $z(s_r)$ is the upper endpoint of $s_r$. Since $C_r$ is to the right of $C$, $z(s_r)$ must be on the left side of $\tau$. We define the three regions $\tau(s_r)$, $\tau(s_r')$, and $\tau(s_r,s_r')$ in the same way as in the proof of Observation~\ref{obser:30}. In this case, all points of $\tau(s_r)$ are below $z(s_r)$ while all points of $\tau(s_r')$ are above $z(s_r)$. With this property, using argument similar to Observation~\ref{obser:30}, the observation can be proved.
\end{proof}

Based on the above two observations, our algorithms for counting type (3.1.2.1) and type (3.1.2.2) intersections are given in the following two lemmas, respectively.

\begin{lemma}
The number of type (3.1.2.1) intersections can be computed in $O(m_{\tau}^2/\log m_{\tau}+n_{\tau}\log^2 m_{\tau})$ time.
\end{lemma}
\begin{proof}
We assume that $C_r$ is below $C$ and the other case where $C_r$ is right of $C$ can be handled similarly. Our algorithm will be based on Observation~\ref{obser:30}(1).
The algorithm is similar to Lemma~\ref{lem:70} but with an additional level data structure. In what follows, we follow the notation in the proof of Lemma~\ref{lem:70} and only point out the differences.
For notational convenience, we let $n=n_{\tau}$ and $m=m_{\tau}$ in this proof.

We follow the same algorithm as in Lemma~\ref{lem:70} but set $r=m/\log^2m$. Recall that for each cell $\sigma\in \Xi_i$, $i=1,2,\ldots,k$, we have defined a subset $L(\sigma)\subseteq L$ consisting of the lunes of $L$ that contain $\sigma$ but do not contain the parent cell of $\sigma$ in $\Xi_{i-1}$.
The total size of all these subsets is $O(mr)$. For each subset $L(\sigma)$, for each $lune(s_b)\in L(\sigma)$, since $lune(s_b)$ contains $\sigma$, all unit disks centered in $\sigma$ contain exactly the same endpoint of $s_b$ and let $p(s_b,\sigma)$ denote the other endpoint; we use $P(\sigma)$ to denote set of all these endpoints $p(s_b,\sigma)$ for all $lune(s_b)\in L(\sigma)$. Note that $|P(\sigma)|=|L(\sigma)|$. We explicitly maintain $P(\sigma)$. Computing these subsets $P(\sigma)$ for all cells can be done in $O(mr)$ time by using the sets $L(\sigma)$. We further sort each set $P(\sigma)$ by their $x$-coordinates. As the total size of $P(\sigma)$ for all cells $\sigma$ is $O(mr)$, the sorting for $P(\sigma)$ of all cells $\sigma$ can be done in $O(mr\log m)$ time.


For each long red arc $s_r$, we first check whether it is a partial arc (this can be done in $O(1)$ time). If not, we ignore it. If yes, we proceed as follows. Without loss of generality, we assume that $s_r$ is to the left of $s_r'$ and thus $z(s_r)$ is the right endpoint of $s_r$. Let $q$ be the center of $s_r$.
We locate the cell $\sigma_i$ of $\Xi_i$ containing $q$ for each $1\leq i\leq k$. For each $\sigma_i$, according to Observation~\ref{obser:30}(1), we perform binary search on the sorted list of $P(\sigma)$ to find the number of points of $P(\sigma)$ to the left of $z(s_r)$, and let $m_{\sigma_i}$ denote this number. We add all these values $m_{\sigma_i}$ to a total count $m_{s_r}$ (initially, $m_{s_r}=0$). Finally, for the cell $\sigma\in \Xi_k$, for each arc $s'\in H_{\sigma}$, we do the following. Note that $s'$ is an arc of $lune(s_b)$ for a short blue arc $s_b$. We check whether $s_r$ and $s_b$ form a type (3.1.2.1) intersection, which can be done in $O(1)$ time; if yes, we increase $m_{s_r}$ by one.
The value $m_{s_r}$ thus obtained is equal to the number of type (3.1.2.1) intersections involving $s_r$. The time for computing $m_{s_r}$ is $O(\log^2 m)$. We apply the above algorithm for each long red arc, after which the number of type (3.1.2.1) intersections is computed. The total time of the algorithm is $O(mr\log m+n\log^2 m)$, which is $O(m^2/\log m+n\log^2 m)$ as $r=m/\log^2 m$.
\end{proof}

\begin{lemma}\label{lem:90}
The number of type (3.1.2.2) intersections can be computed in $O(m_{\tau}^2/\log m_{\tau}+n_{\tau}\log^2 m_{\tau})$ time.
\end{lemma}
\begin{proof}
We assume that $C_r$ is below $C$ and the other case where $C_r$ is right of $C$ can be handled similarly. Our algorithm will be based on Observation~\ref{obser:30}(2). For notational convenience, let $n=n_{\tau}$ and $m=m_{\tau}$ in this proof.

We still follow the algorithm in the Lemma~\ref{lem:70} but set $r=m/\log^2 m$ and use $lune'(s_b)$ to replace $lune(s_b)$. Note that the intersection of $C_r$ and $lune'(s_r)$ is still bounded by at most two circular arcs (other than the boundary of $C_r$). Let $H$ be the set of all these circular arcs for all short blue arcs. Let $L$ denote the set of $lune'(s_b)$ of all short blue arcs $s_b$.
Next, we build a hierarchical cutting for $H$ in the same way as in the proof of Lemma~\ref{lem:70}.
We define $L(\sigma)$ in the same way. Let $B(\sigma)$ denote the set of short blue arcs $s_b$ with $lune'(s_b)\in L(\sigma)$. Using $L(\sigma)$, we can obtain $B(\sigma)$ for all cells $\sigma\in \Xi_i$, $i=1,2,\ldots,k$, in $O(mr)$ time. For each arc $s_b\in B(\sigma)$, it defines an interval $[x_1(s_b),x_2(s_b)]$, where $x_1(s_b)$ and $x_2(s_b)$ are the $x$-coordinates of the left and right endpoints of $s_b$, respectively; let $I(\sigma)$ denote the set of all these intervals for all arcs of $B(\sigma)$. We build an interval tree (Section~10.1~\cite{ref:deBergCo08}) on the intervals of $I(\sigma)$ in $O(|I(\sigma)|\log |I(\sigma)|)$ time so that given a query $x$-coordinate $x$, we find in $O(\log |I(\sigma)|)$ time the number of intervals of $I(\sigma)$ containing $x$. As $|I(\sigma)|=O(m)$ and the total size of the sets $I(\sigma)$'s of all cells in the cutting is $O(mr)$, the total time for building the interval trees is $O(mr\log m)$.

For each long red arc $s_r$, we first check whether it is a partial arc (this can be done in $O(1)$ time). If not, we ignore it. If yes, we proceed as follows. Without loss of generality, we assume that $s_r$ is to the left of $s_r'$ and thus $z(s_r)$ is the right endpoint of $s_r$. Let $q$ be the center of $s_r$.
We locate the cell $\sigma_i$ of $\Xi_i$ containing $q$ for each $i=1,2,\ldots,k$. For each $\sigma_i$, we query the interval tree of $I(\sigma_i)$ to find number of intervals of $I(\sigma_i)$ containing $z(s_r)$; let $m_{\sigma_i}$ denote this number. We add all these values $m_{\sigma_i}$ to a total count $m_{s_r}$ (initially, $m_{s_r}=0$). Finally, for the cell $\sigma\in \Xi_k$, for each arc $s'\in H_{\sigma}$, we do the following. Let $s_b$ be the short-blue arc such that $s'$ is an arc of $lune'(s_b)$.  We check whether $s_r$ and $s_b$ form a type (3.1.2.1) intersection, which can be done in $O(1)$ time; if yes, we increases $m_{s_r}$ by one.
The value $m_{s_r}$ thus obtained is equal to the number of type (3.1.2.2) intersections involving $s_r$. The time for computing $m_{s_r}$ is $O(\log^2 m)$. We apply the above algorithm for each long red arc, after which the number of type (3.1.2.2) intersections is computed. The total time of the algorithm is $O(mr\log m+n\log^2 m)$, which is $O(m^2/\log m+n\log^2 m)$ as $r=m/\log^2 m$.
\end{proof}

In summary, the number of type (3.1.2) intersections can be computed in $O(m_{\tau}^2/\log m_{\tau}+n_{\tau}\log^2 m_{\tau})$ time.

\subsubsection{Counting type (3.2) intersections}

We now compute the number of type (3.2) intersections. The runtime of the algorithm is $O(m_{\tau}^2\log^2m_{\tau}+n_{\tau}\log^4 m_{\tau})$.
Our goal is to compute the number of pairs of a long red arc $s_r$ and a short blue arc $s_b$ such that $s_r$ and $s_b$ intersect twice.

\begin{figure}[t]
\begin{minipage}[t]{\textwidth}
\begin{center}
\includegraphics[height=1.4in]{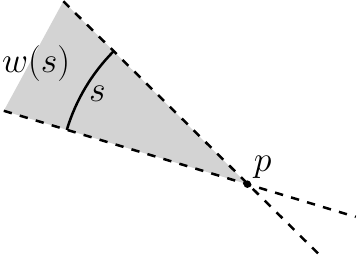}
\caption{\footnotesize Illustrating the wedge $w(s)$ (the grey area): $p$ is the center of $s$.}
\label{fig:wedge}
\end{center}
\end{minipage}
\vspace{-0.15in}
\end{figure}


Let $s$ be either a blue short arc or a long red arc. Since the radius of $s$ is $1$, $s$ is contained in $\tau$, which is in $C$, and $C$ is a square cell of side-length $1/\sqrt{2}$, $s$ does not span more than a semicircle of the underlying circle $\alpha(s)$ of $s$.
The two lines through the center of $s$ and its two endpoints partition the plane into four wedges, one of which contains $s$ completely (we use $w(s)$ to denote the wedge; e.g., see Fig.~\ref{fig:wedge}). The following observation is critical to our algorithm in Lemma~\ref{lem:80}.

\begin{figure}[t]
\begin{minipage}[t]{\textwidth}
\begin{center}
\includegraphics[height=3.0in]{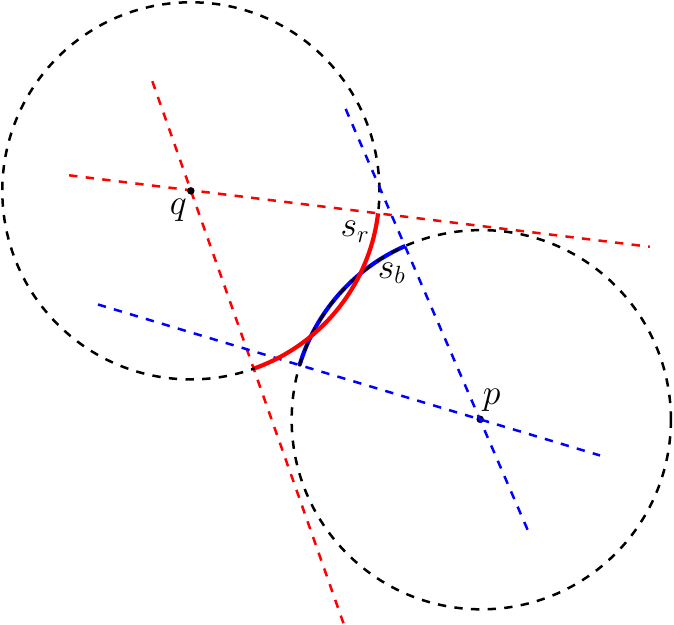}
\caption{\footnotesize Illustrating Observation~\ref{obser:50}: $p$ is the center of $s_b$ and $q$ is the center of $s_r$.
}
\label{fig:obser30}
\end{center}
\end{minipage}
\vspace{-0.15in}
\end{figure}

\begin{observation}\label{obser:50}
Consider a short blue arc $s_b$ and a long red arc $s_r$. If $s_b$ and $s_r$ intersect twice, then the  following four conditions must hold (e.g., see Fig.~\ref{fig:obser30}): (1) the center of $s_r$ is in the wedge $w(s_b)$; (2) the center of $s_b$ is in the wedge $w(s_r)$; (3) the underlying disk $D(s_r)$ of $s_r$ does not contain either endpoint of $s_b$; (4) the underlying circles $\alpha(s_b)$ and $\alpha(s_r)$ intersect.
On the other hand, if the above four conditions all hold, then exactly one of the following two cases must happen: (1) $s_b$ and $s_r$ intersect twice; (2) $s_r$ is a partial arc and $s_b$ intersects both $s_r$ and its coupled arc $s_r'$.
\end{observation}
\begin{proof}
First of all, we have the following observation from \cite[Lemmas~3.2]{ref:AgarwalCo93}.

\begin{figure}[t]
\begin{minipage}[t]{\textwidth}
\begin{center}
\includegraphics[height=3.0in]{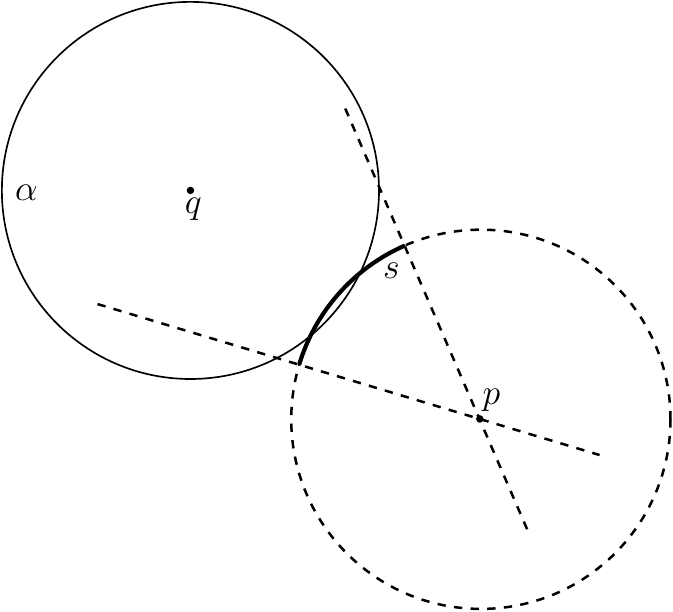}
\caption{\footnotesize Illustrating Observation~\ref{obser:60}: $p$ is the center of $\alpha$ and $q$ is the center of $s$.}
\label{fig:claim}
\end{center}
\end{minipage}
\vspace{-0.15in}
\end{figure}

\begin{observation}\label{obser:60}{\em (Agarwal, Pellegrini, and Sharir~\cite{ref:AgarwalCo93})}
Suppose $s$ is a unit arc that does not span more than a semicircle and $\alpha$ is a unit circle.
Then $s$ intersects $\alpha$ twice if and only if the following three conditions all hold (e.g. see Fig.~\ref{fig:claim}): (1) the center of $\alpha$ is in the wedge $w(s)$; (2) $\alpha$ does not contain either endpoint of $s$; (3) $\alpha$ intersects the underlying circle of $s$.
\end{observation}



We now prove the observation using Observation~\ref{obser:60}. First of all, if $s_b$ and $s_r$ intersect twice, then according to Observation~\ref{obser:60}, all four conditions in the observation must hold.
Suppose that all four conditions hold for $s_b$ and $s_r$. We argue below that either $s_b$ and $s_r$ intersect twice, or $s_r$ is a partial arc and $s_b$ intersects both $s_r$ and $s_r'$.
According to Observation~\ref{obser:60}, conditions (1-3) guarantee that $s_b$ intersects $\alpha(s_r)$ twice, say, at two points $a$ and $a'$. Since $s_b\in C$, both $a$ and $a'$ are in $C$.

\begin{itemize}
\item
If $s_r$ is not a partial arc, then $s_r=\alpha(s_r)\cap C$. Hence, both $a$ and $a'$ are on $s_r$, implying that $s_b$ and $s_r$ intersect twice (and $s_b$ cannot intersect $s_r'$).

\item
If $s_r$ is a partial arc, then if both $a$ and $a'$ are on $s_r$, $s_b$ and $s_r$ intersect twice (and $s_b$ cannot intersect $s_r'$). Now assume that not both $a$ and $a'$ are on $s_r$. Then, either only one of $a$ and $a'$ is on $s_r$, or neither one is on $s_r$.

\begin{itemize}
  \item If only one of $a$ and $a'$ is on $s_r$, say $a\in s_r$, then $a'$ is on $s_r'$ since $a'\in C$ and $s_r\cup s_r'=\alpha(s_r)\cap C$. Hence, $s_b$ intersects both $s_r$ and $s_r'$.
  \item If neither $a$ nor $a'$ is on $s_r$, then both $a$ and $a'$ are on $s_r'$. According to Observation~\ref{obser:60}, the center of $s_b$ must be in the wedge $w(s_r')$. Note that $w(s_r)$ and $w(s_r')$ are disjoint. Hence, the center of $s_b$ cannot be in $w(s_r)$. But this contradicts with the four condition. As such, the case where neither $a$ nor $a'$ is on $s_r$ cannot happen.
\end{itemize}
\end{itemize}

This proves the observation.
\end{proof}

Notice that in the second case of Observation~\ref{obser:50} $s_r$ and $s_b$ form a type (3.1.2.2) intersection. Therefore, if $k_1$ is the number of pairs of a long red arc $s_r$ and a short blue arc $s_b$ that intersect twice, then $k_1=k_2-k_3$, where $k_2$ is the number of pairs $(s_r,s_b)$ that satisfy the four conditions in Observation~\ref{obser:50} and $k_3$ is the number of type (3.1.2.2) intersections. We have already computed $k_3$ in Lemma~\ref{lem:90}. It remains to compute $k_2$, which is done in the following lemma.

\begin{lemma}\label{lem:80}
The number of type (3.2) intersections can be computed in $O(m_{\tau}^2\log^2 m_{\tau}+n_{\tau}\log^4 m_{\tau})$ time.
\end{lemma}
\begin{proof}
We first apply Lemma~\ref{lem:90} and compute $k_3$ in $O(m_{\tau}^2/\log m_{\tau}+n_{\tau}\log^2 m_{\tau})$ time. As discussed before, it suffices to compute $k_2$, i.e., the number of pairs $(s_r,s_b)$ that satisfy the four conditions in Observation~\ref{obser:50}. We show below that $k_2$ can be computed in $O(m_{\tau}^2\log^2 m_{\tau}+n_{\tau}\log^4 m_{\tau})$ time.
For notational convenience, we let $n=n_{\tau}$ and $m=m_{\tau}$ in this proof.

Let $B$ denote the set of short blue arcs  and let $R$ be the set of long red arcs of $S(\tau)$.
For each short blue arc $s_b$, define $A(s_b)$ as the
common intersection of the wedge $w(s_b)$ and $lune'(s_b)$.
Note that the underlying disk of a long-red arc $s_r$ does not contain either endpoint of a short-blue arc $s_b$ if and only if the center of $s_r$ is in $lune'(s_b)$. Hence, the first two conditions of Observation~\ref{obser:50} are satisfied if and only if the center of $s_r$ is in $A(s_b)$.

Recall that centers of all red arcs are in the square cell $C_r$.
Let $A'(s_b)=A(s_b)\cap C_r$. It is not difficult to see that other
than those on $\partial C_r$,
$\partial A'(s_b)$ consists of at most two line segments and at most two
circular arcs; we use {\em general-arcs} to refer to these arcs and segments (a segment can be considered
as a circular arc of infinite radius).
Hence, $\partial A'(s_b)$  has at most four general-arcs. We call
$A'(s_b)$ the {\em interesting region} of $s_b$. Let $I$ denote the
set of interesting regions of all short blue segments of $B$. Let $E$
denote the set of general-arcs of all regions of $I$. Hence, $|E|=O(m)$.

We compute a hierarchical $(1/r)$-cutting $\Xi_0,\Xi_1,\ldots,\Xi_k$
on the general-arcs of $E$, with $r=m$ and  a constant ratio $\rho$ as described
in Section~\ref{sec:pre}. This can be done in $O(mr)$ time by Theorem~\ref{theo:cutting}.
For each cell $\sigma\in \Xi_i$, $1\leq i\leq k$, we define a canonical pair $(B_{\sigma},R_{\sigma})$, with $B_{\sigma}\subseteq B$ and $R_{\sigma}\subset R$, so that centers of all arcs of $R_{\sigma}$ lie in the interesting region of each short blue arc of $B_{\sigma}$. Specifically, $R_{\sigma}$ consists of all long red arcs whose centers are located in the cell $\sigma$ and $B_{\sigma}$ consists of the arcs of $B$ whose interesting regions contain $\sigma$ but do not contain the parent cell of $\sigma$ in $\Xi_{i-1}$. The canonical pairs of all cells of all cuttings $\Xi_i$, $1\leq i\leq k$, can be computed in $O(mr + n\log m)$ time, as follows.

For the center $p$ of each long red arc $s_r$, we locate the cell $\sigma$ of $\Xi_i$ containing $p$, for all $1\leq i\leq k$, and add $s_r$ to $R_{\sigma}$, which can be done in $O(\log r)$ time. As such, computing the canonical sets $R_{\sigma}$ for all cells $\sigma$ takes $O(n\log m)$ time in total. For computing $B_{\sigma}$, we can use an algorithm similar to the one in Lemma~\ref{lem:70} (i.e., replacing $lune(s_b)$ by the interesting region $A'(s_b)$).
More specifically, for each child cell $\sigma\in \Xi_i$ of a cell $\sigma'$ in $\Xi_{i-1}$, for each general-arc of $E_{\sigma'}$ (i.e., the set of general-arcs of $E$ crossing $\sigma'$), suppose it is a general-arc of $A'(s_b)$ of a blue arc $s_b$. We check whether $A'(s_b)$ contains $\sigma$; if yes, we add $s_b$ to $B_{\sigma}$. In this way, computing all canonical subsets $B_{\sigma}$ takes time proportional to the total size of the subsets $E_{\sigma'}$ for all cells $\sigma'$ of all cuttings $\Xi_i$, $i=0,1,\ldots,k-1$, which is $O(mr)$.

By the definition of canonical subset pairs $(B_{\sigma},R_{\sigma})$, to compute $k_2$, it suffices to compute the number of pairs of arcs $(s_b,s_r)$ such that $s_b\in B_{\sigma}$, $s_r\in R_{\sigma}$, $\alpha(s_b)$ intersects $\alpha(s_r)$, and the center of $s_b$ is in $w(s_r)$, for all cells $\sigma\in \Xi_i$, $i=1,2,\ldots,k$. In the following, we describe an algorithm of $O(m_{\sigma}^2/\log m+n_{\sigma}\log m)$ time for one such cell $\sigma\in \Xi_i$, where $m_{\sigma}=|B_{\sigma}|$ and $n_{\sigma}=|R_{\sigma}|$.

Recall that all centers of red arcs are in the square cell $C_r$. Define $\calD$ as the set of disks of radius $2$ centered at the centers of all blue arcs of $B_{\sigma}$. Let $P$ be the set of all centers of the arcs of $R_{\sigma}$; if $p\in P$ is the center of an arc $s_r\in R_{\sigma}$, we use $s(p)$ to refer to $s_r$. It is not difficult to see that the problem is equivalent to computing the number of pairs $(p,D)$, $p\in P$, $D\in \calD$, such that $p$ is contained in $D$ and the center of $D$ is in $w(s(p))$. To solve the problem, we use an algorithm similar to that for Lemma~\ref{lem:70}, with an additional level of range searching data structure for handling the constraint that the center of $D$ is in $w(s(p))$. We briefly discuss it below. Note that $m_{\sigma}=|\calD|$ and $n_{\sigma}=|P|$.

For each disk $D\in \calD$, the intersection of $\partial D$ and $C_r$ consists of at most two circular arcs.
Let $H$ denote the set of circular arcs $\partial D\cap C_r$ for all disks $D\in \calD$.
We build a hierarchical $(1/r')$-cutting $\Xi'_0,\Xi'_1,\ldots,\Xi'_{k'}$ for $H$, with $r'= m_{\sigma}/\log m$ and a constant ratio $\rho_1$, in $O(m_{\sigma}r')$ time~\cite{ref:ChazelleCu93}. For each cell $\beta\in \Xi'_i$, for all $i=1,2,\ldots,k$, define $\calD(\beta)$ as the set of disks of $\calD$ that contain $\beta$ but do not contain the parent cell of $\beta$ in $\Xi'_{i-1}$. We can compute $\calD(\beta)$ explicitly for all cells of in $O(m_{\sigma}r')$ time. We further build a range searching data structure on the centers of $\calD(\beta)$ for answering the following queries: Given a query wedge, compute the number of disk centers in the wedge. For this, we use a data structure of Matou\v{s}ek's~\cite{ref:MatousekRa93} (i.e., Lemma~6.1~\cite{ref:MatousekRa93} by setting $p=2$), which takes $O(|\calD(\beta)|^2/\log^{1-\delta} |\calD(\beta)|)$ preprocessing time for any $\delta>0$ and can answer each query in $O(\log^2 |\calD(\beta)|)$ time.
In addition, for each cell $\beta\in \Xi'_{k'}$, we explicitly store the subset $H_{\beta}$ of arcs of $H$ crossing $\beta$.

All above can be done in $O(m_{\sigma}r')$ time except that for building the wedge range searching data structures. We now analyze the total time for constructing these wedge range searching data structures. For each cell $\beta\in \Xi_i$, it holds that $|\calD(\beta)|=O(m_{\sigma}/\rho_1^{i-1})$. Since $\Xi_i$ has $O(\rho_1^{2i})$ cells, the total time for constructing the data structures is big-O of $\sum_{i=1}^{k'}\rho_1^2\cdot m_{\sigma}^{2i}/\rho_1^{2(i-1)}/\log^{1-\delta} (m_{\sigma}/\rho_1^{i-1})$, which is bounded by $O(m_{\sigma}^2\log m)$ as $\rho_1$ is a constant.

For each point $p\in P$, we locate the cell $\beta_i$ of $\Xi'_i$ containing $p$ for each $i=1,2,\ldots,k'$;
using the wedge range searching data structure on the centers of $\calD(\beta_i)$, we find the number $m_{\beta_i}$ of centers of disks of $\calD(\beta_i)$ contained in the wedge $w(s(p))$. We add all these values $m_{\beta_i}$ to a total count $m_p$ (initially, $m_p=0$). Finally, for the cell $\beta_k\in \Xi'_k$, for each arc $s\in H_{\beta_k}$, we check whether the underlying disk of $s$ contains $p$ and the center of $s$ is in $w(s)$; if yes, we increases $m_p$ by one.
The value $m_p$ thus obtained is equal to the number of disks $D$ of $\calD$ that contain $p$ and whose centers are in $w(s(p))$. The time for computing $m_p$ is $O(\log r'\log^2 m)$ as $|H_{\beta}|=O(\log m)$ for any cell $\beta\in \Xi'_k$.
The sum $\sum_{p\in P}{m_p}$ is equal to the number of pairs $(p,D)$, $p\in P$, $D\in \calD$, such that $p$ is contained in $D$ and the center of $D$ is in $w(s(p))$. As $n_{\sigma}=|P|$, the total time of the algorithm is $O(m_{\sigma}r'+m_{\sigma}^2\log m+n_{\sigma}\log r'\log^2 m))$, which is $O(m_{\sigma}^2\log m+n_{\sigma}\log^3 m)$ for $r'= m_{\sigma}/\log m$.

The above shows that processing each canonical subset pair $(B_{\sigma},R_{\sigma})$ for a cell $\sigma\in \Xi_i$ takes $O(m_{\sigma}^2\log m+n_{\sigma}\log^3 m)$ time, with $m_{\sigma}=|B_{\sigma}|$ and $n_{\sigma}=|R_{\sigma}|$. Processing all cells $\sigma\in \Xi_i$ for all $i=1,2,\ldots,k$ will compute the number $k_2$. We now analyze the total time.
For each $1\leq i\leq k$, $\sum_{\sigma\in \Xi_i}|R_{\sigma}|=n$ and $m_{\sigma}\leq |E_{\sigma'}|$, where $\sigma'$ is the parent cell of $\sigma$ in $\Xi_{i-1}$. Since $|E_{\sigma'}|=O(m/\rho^{i-1})$, we have $m_{\sigma}=O(m/\rho^{i-1})$. Because $\Xi_i$ has $O(\rho^{2i})$ cells, the total time for processing all cells in $\Xi_i$ is on the order of $\rho^{2i}\cdot m_{\sigma}^2\log m+n\log^3 m\leq \rho^{2i}\cdot m^2/\rho^{2i-2}\cdot \log m+n\log^3 m$, which is bounded by $O(m^2\log m+n\log^3 m)$ as $\rho$ is a constant. Since $k=O(\log r)$ and $r=m$, the total time for processing all cells $\sigma\in \Xi_i$ for all $i=1,2,\ldots,k$ is $O(m^2\log^2 m+n\log^4 m)$.
\end{proof}


\subsubsection{Counting type (3) intersections: a final step}

The above shows that the number of type (3) intersections can be
computed in $O(m_{\tau}^2\log^2m_{\tau}+n_{\tau}\log^4 m_{\tau})$ time.
Using the result, we finally show that an alternative algorithm can compute the number of type (3) intersections  in $O(n_{\tau}\log^4
m_{\tau}+m_{\tau}\sqrt{n_{\tau}}\log^3 m_{\tau})$ time.

If $m_{\tau}<\sqrt{n_{\tau}}\log m_{\tau}$, then
$m_{\tau}^2\log^2m_{\tau}\leq n_{\tau}\log^4 m_{\tau}$ and
thus the runtime of the above algorithm is bounded by $O(n_{\tau}\log^4 m_{\tau})$.
Otherwise, we partition the short blue arcs into $t$ groups of size
$m_{\tau}/t$, with $t=\frac{m_{\tau}}{\sqrt{n_{\tau}}\log m_{\tau}}$.
Then, for each group, we apply the above algorithm on the group with
all long-red arcs. The total time of the algorithm is $O(t\cdot
(m_{\tau}^2/t^2\log^2(m_{\tau}/t) + n_{\tau}\log^4(m_{\tau}/t)))$, which is
$O(m_{\tau}\sqrt{n_{\tau}}\log^3 m_{\tau})$. As such, the number of type
(3) intersections can be computed in $O(n_{\tau}\log^4
m_{\tau}+m_{\tau}\sqrt{n_{\tau}}\log^3 m_{\tau})$ time.

\subsection{Counting type (1) intersections}
\label{sec:type1}

Before giving our algorithm for counting the type (1) intersections, we analyze the total time for counting all intersections of other types in our original problem. Recall that initially we build a hierarchical $(1/r)$-cutting $\Xi_0,\Xi_1,\ldots,\Xi_k$ on the arcs of $S'$ with $n=|S'|$. For each cell $\tau$ of $\Xi_k$, we have the pseudo-trapezoid-restricted subproblem of counting the intersections of arcs of $S(\tau)$ inside $\tau$, with $n_{\tau}=|S(\tau)|$ and $m_{\tau}$ as the number of short arcs of $S(\tau)$.
By the definition of short and long arcs, we have $m_{\tau}\leq n_{\tau}=O(n/r)$ and $\sum_{\tau\in \Xi_k}m_{\tau}=O(n)$. By setting $r=n^{1/3}/(\log n)^{2/3}$, we obtain the following lemma.


\begin{lemma}\label{lem:105}
By setting $r=n^{1/3}/(\log n)^{2/3}$, the total time for computing the number of type (2), (3), and (4) intersections for all cells $\tau\in \Xi_k$ is bounded by $O(n^{4/3}(\log n)^{10/3})$.
\end{lemma}
\begin{proof}
To solve the  pseudo-trapezoid-restricted subproblem on each cell $\tau\in \Xi_k$, we have shown above that counting type (2) intersections can be done in $O(m_{\tau}^{4/3+\epsilon})$ time and counting types (3) and (4) intersections takes $O(n_{\tau}\log^4m_{\tau}+m_{\tau}\sqrt{n_{\tau}}\log^3 m_{\tau})$ time.
Therefore, the total time for counting type (2) intersections for all cells $\tau\in\Xi_k$ is $O(\sum_{\tau\in \Xi_k} m_{\tau}^{4/3+\epsilon})$. Since $\Xi_k$ has $O(r^2)$ cells, $\sum_{\tau\in \Xi_k}m_{\tau}=O(n)$, and $m_{\tau}=O(n/r)$, $\sum_{\tau\in \Xi_k} m_{\tau}^{4/3+\epsilon}$ achieves the maximum when $r$ $m_{\tau}$'s have value $n/r$, i.e., $\sum_{\tau\in \Xi_k} m_{\tau}^{4/3+\epsilon}=O(r\cdot (n/r)^{4/3+\epsilon})$, which is $O(n^{4/3+\epsilon}/r^{1/3+\epsilon})$. Hence, the total time for counting type (2) intersections is $O(n^{4/3+\epsilon}/r^{1/3+\epsilon})$.

The total time for counting types (3) and (4) intersections is $\sum_{\tau\in \Xi_k} (n_{\tau}\log^4m_{\tau}+m_{\tau}\sqrt{n_{\tau}}\log^3 m_{\tau})$.
Since $\Xi_k$ has $O(r^2)$ cells and $n_{\tau}=O(n/r)$, we have $\sum_{\tau\in \Xi_k} n_{\tau}\log^4m_{\tau}=O(r^2\cdot n/r\cdot \log^4 n)$, which is $O(nr\log^4 n)$. Since $\sum_{\tau\in \Xi_k}m_{\tau}=O(n)$, we can derive $\sum_{\tau\in \Xi_k} m_{\tau}\sqrt{n_{\tau}}\log^3 m_{\tau}=O(n\sqrt{\frac{n}{r}}\log^3 n)$.

As such, the total time for computing the number of intersections of types (2-4) is $O(n^{4/3+\epsilon}/r^{1/3+\epsilon}+nr\log^4 n+n\sqrt{\frac{n}{r}}\log^3 n)$. Setting $r=n^{1/3}/(\log n)^{2/3}$ makes the time bounded by $O(n^{4/3}(\log n)^{10/3})$.
\end{proof}

We now discuss how to compute the number of type (1) intersections, i.e., intersections between long red arcs and long blue arcs. If we consider each individual cell of the cutting $\Xi_k$ as above, it would be difficult to achieve a satisfying time bound because each long arc may intersect many cells.
Instead, we will count these intersections using the cuttings $\Xi_0,\Xi_1,\ldots,\Xi_k$ in the hierarchy. For each cell $\sigma\in \Xi_i$ for any $0\leq i\leq k$, we define long and short arcs of $\sigma$ in the same way as before (i.e., suppose an arc $s\in S$ intersects the interior of $\sigma$; then
$s$ is a {\em long arc of $\sigma$} if $s$ does not have either endpoint in the interior of $\sigma$ and is a {\em short arc of $\sigma$} otherwise).
For a long arc $s$ of $\sigma$, we say that $s$ is a {\em long-long arc} if $i>0$ and $s$ is also a long arc of $\sigma'$, where $\sigma'$ is the parent cell of $\sigma$ in $\Xi_{i-1}$, and $s$ is a {\em short-long arc} otherwise.
A critical observation is that for a long arc $s$ of any cell $\sigma\in \Xi_k$, $s$ must be a short-long arc of exactly one ancestor cell of $\sigma$ since any arc must be a short arc of the only cell of $\Xi_0$, which is the entire square cell $C$.

Recall that a type (1) intersection refers to an intersection in a cell $\sigma\in \Xi_k$ between a long red arc of $\sigma$ and a long blue arc of $\sigma$.
Consider a type (1) intersection $q$ in a cell $\sigma\in \Xi_k$ between a long red arc $s_r$ and a long blue arc $s_b$. Observe that $\sigma$ has one and only one ancestor cell $\sigma'$ such that both $s_r$ and $s_b$ are long arcs of $\sigma'$ but at least one of $s_r$ and $s_b$ is a short-long arc of $\sigma'$ (also note that since $\sigma'$ is an ancestor cell of $\sigma$, $\sigma'$ contains $\sigma$ and thus contains $q$). This means that in order to count type (1) intersections, we can count, for all cells $\sigma\in \Xi_i$, $i=0,1,\ldots,k$, the type (1) intersections in $\sigma$ involving at least one short-long arc, and more precisely, we can count the following three types of intersections:
(1.1) intersections between long-long red arcs of $\sigma$ and short-long blue arcs of $\sigma$; (1.2) intersections between long-long blue arcs of $\sigma$ and short-long red arcs of $\sigma$; (1.3) intersections between short-long blue arcs of $\sigma$ and short-long red arcs of $\sigma$.

Consider a cell $\sigma\in \Xi_i$, for any $0\leq i\leq k$. Let $S_r$ be the set of all long red arcs of $\sigma$ and $S_b$ the set of all long blue arcs of $\sigma$. Let $S^1_r$ and $S^2_r$ be the subsets of long-long and short-long arcs of $\sigma$ in $S_r$, respectively. Let $S^1_b$ and $S^2_b$ be the subsets of long-long and short-long arcs of $\sigma$ in $S_b$, respectively. Let $n_{\sigma}=|S_r|+|S_b|$ and $m_{\sigma}=|S_r^2|+|S_b^2|$.
As such, the above three types of intersections are: (1.1) intersections between arcs of $S^1_r$ and $S^2_b$, (1.2) intersections between arcs of $S^2_r$ and $S^1_b$, and (1.3) intersections between arcs of $S^2_r$ and $S^2_b$.

To compute the number of type (1.1) intersections, one could apply the algorithm in Section~\ref{sec:type3}. However, the total time of the overall algorithm for all cells of $\Xi_i$, $i=0,1,\ldots,k$, cannot be bounded by $O(n^{4/3}\cdot\text{poly}(\log n))$ since now we need to consider the cells in all cuttings $\Xi_i$, $i=0,1,\ldots,k$, not just the last cutting $\Xi_k$ as before. Instead, we present a new algorithm in Lemma~\ref{lem:110}.
Type (1.2) intersections can be handled similarly as it is symmetric to type (1.1). For type (1.3), we actually also apply Lemma~\ref{lem:110} because the algorithm works for any two sets of long arcs of $\sigma$, and the running time is also bounded as stated in Lemma~\ref{lem:110} due to $m_{\sigma}\leq n_{\sigma}$.


\begin{lemma}\label{lem:110}
There exists an algorithm of $O(m_{\sigma}\log^4 n_{\sigma}+n_{\sigma}\log^5 n_{\sigma}+n_{\sigma}^{2/3}m_{\sigma}^{2/3}(\log n_{\sigma})^{13/3})$ time to compute the number of type (1.1) intersections, i.e., intersections between arcs of $S^1_r$ and arcs of $S^2_b$.
\end{lemma}

Since the proof of Lemma~\ref{lem:110} is lengthy and technical, we devote Section~\ref{sec:lemproof} to it.
Using Lemma~\ref{lem:110}, the following lemma analyzes the total time of our algorithm for computing the number of type (1) intersections.

\begin{lemma}\label{lem:130}
The number of type (1) intersections can be computed in $O(n^{4/3}\log^{16/3} n)$ time.
\end{lemma}
\begin{proof}
As discussed before, the number of type (1.1) intersections can be computed once the subproblems stated in Lemma~\ref{lem:110} are solved for all cells $\sigma\in \Xi_i$, $i=0,1,\ldots,k$.
We next analyze the total time for solving these subproblems. We follow the notation as defined before.
Note that $n_{\sigma}=O(\frac{n}{\rho^i})$ and $\Xi_i$ has $O(\rho^{2i})$ cells.

We first claim that $\sum_{\sigma\in \Xi_i}m_{\sigma}=O(n)$ for all $0\leq i\leq k$. Indeed, it is obviously true that $\sum_{\sigma\in \Xi_i}m_{\sigma}=O(n)$ for $i=0$. We assume that $i>0$.
Consider a short-long arc $s$ of $\sigma$. By definition, $s$ is a short arc of the parent cell $\sigma'\in \Xi_{i-1}$ of $\sigma$, implying that $s$ has at least one endpoint in the interior of $\sigma'$. Let $\sigma_1$ and $\sigma_2$ be the two cells of $\Xi_{i-1}$ that contain the two endpoints of $s$, respectively ($\sigma_1=\sigma_2$ is possible). According to the above discussion, if $s$ is short-long arc of $\sigma$ for some cell $\sigma\in \Xi_i$, then $\sigma$ must be a child cell of either $\sigma_1$ or $\sigma_2$. As $\sigma_j$ has $O(1)$ cells in $\Xi_i$, for $j=1,2$, $s$ can be a short-long arc of at most $O(1)$ cells $\sigma\in \Xi_i$. As such, $\sum_{\sigma\in \Xi_i}m_{\sigma}=O(n)$ holds.

We now consider the three terms in the time complexity of
Lemma~\ref{lem:110} separately.

\begin{enumerate}
  \item For the first term, $\sum_{i=0}^{k}\sum_{\sigma\in \Xi_i}m_{\sigma}\log^4 n_{\sigma}=O(\sum_{i=1}^{k}n\log^4 n)=O(n\log^5n)$, for $k=O(\log r)$ and $r=n^{1/3}/(\log n)^{2/3}$.
  \item For the second term, we have $\sum_{i=0}^{k}\sum_{\sigma\in \Xi_i}n_{\sigma}\log^5 n_{\sigma}=O(\sum_{i=1}^{k}\rho^{2i}\cdot \frac{n}{\rho^i}\log^5 n)$, which is bounded by $O(nr\log^5 n)$ as $r=\Theta(\rho^k)$.
  \item For the third term, $\sum_{\sigma\in \Xi_i} n_{\sigma}^{2/3}m_{\sigma}^{2/3}(\log n_{\sigma})^{13/3}$ is big-O of $\sum_{\sigma\in \Xi_i} (\frac{n}{\rho^i})^{2/3}m_{\sigma}^{2/3}(\log n)^{13/3}=(\frac{n}{\rho^i})^{2/3}\cdot (\log n)^{13/3}\cdot \sum_{\sigma\in \Xi_i} m_{\sigma}^{2/3}$. As $\sum_{\sigma\in \Xi_i}m_{\sigma}=O(n)$ and $\Xi_i$ has $O(\rho^{2i})$ cells, by H\"older's Inequality, $\sum_{\sigma\in \Xi_i} m_{\sigma}^{2/3}=O(\rho^{2i}\cdot (\frac{n}{\rho^{2i}})^{2/3})$. As such, we obtain that $\sum_{\sigma\in \Xi_i} n_{\sigma}^{2/3}m_{\sigma}^{2/3}(\log n_{\sigma})^{13/3}$ is big-O of $(\frac{n}{\rho^i})^{2/3}\cdot (\log n)^{13/3}\cdot\rho^{2i}\cdot (\frac{n}{\rho^{2i}})^{2/3}$, which is $O(n^{4/3}(\log n)^{13/3})$ as $\rho$ is a constant. Hence $\sum_{i=0}^{k}\sum_{\sigma\in \Xi_i} n_{\sigma}^{2/3}m_{\sigma}^{2/3}(\log n_{\sigma})^{13/3}=O(n^{4/3}(\log n)^{16/3})$, for $k=O(\log r)$ and $r=n^{1/3}/(\log n)^{2/3}$.
\end{enumerate}

Therefore, the total time for solving the subproblems for all cells $\sigma\in \Xi_i$ for all $i=0,1,\ldots,k$ is $O(n\log^5 n+nr\log^5 n+n^{4/3}\log^{16/3} n)$, which is $O(n^{4/3}\log^{16/3} n)$ as $r=n^{1/3}/(\log n)^{2/3}$.
Hence, the number of type (1.1) intersections can be computed in $O(n^{4/3}\log^{16/3} n)$ time.

As discussed before, Lemma~\ref{lem:110} is also applicable to counting type (1.2) and type (1.3) intersections with asymptotically the same time complexity. As such, the total number of type (1) intersections can be computed in $O(n^{4/3}\log^{16/3} n)$ time.
\end{proof}

Combining Lemmas~\ref{lem:105} and \ref{lem:130}, we conclude that the total number of intersections all types (i.e., the number of intersections in the square cell $C$) can be computed in $O(n^{4/3}\log^{16/3} n)$ time.
Note that the time complexities in the two lemmas imply that computing type (1) intersections is the bottleneck.

\subsubsection{Proof of Lemma~\ref{lem:110}}
\label{sec:lemproof}

In this section, we prove  Lemma~\ref{lem:110}.
Recall that $|S_b^1|+|S_r^2|\leq n_{\sigma}$ and $|S_r^2|\leq m_{\sigma}$. For notational convenience, we let $n=n_{\sigma}$ and $m=m_{\sigma}$ in this section (and $n_{\sigma}$ and $m_{\sigma}$ will be used for other purposes in this section). Let $B=S_b^1$ and let $R=S_r^2$.

Since we are only interested in the intersections of $B$ and $R$ inside $\sigma$, it suffices to consider the portions of the arcs of $B$ and $R$ in $\sigma$. Note that each arc $s\in B\cup R$ may intersect $\sigma$ at two sub-arcs, each of which is a long arc of $\sigma$ (i.e., both endpoints are on $\partial \sigma$) because $s$ does not have any endpoint in the interior of $\sigma$. For simplicity, we still use $B$ (resp., $R$) to denote the set of the portions of the arcs of $B$ (resp., $R$) inside $\sigma$. Note that the size of $B$ (resp., $R$) may grow by a factor of $2$, but to simplify the notation, we still let $n=|B|+|R|$ and $m=|B|$, which will not affect our time complexity asymptotically.

We partition the intersections between $B$ and $R$ into two types: (1.1.1)
pairs of arcs $(s_b,s_r)$, with $s_b\in B$ and $s_r\in R$, such that $s_b$ and $s_r$ intersect only once; (1.1.2) pairs of arcs $(s_b,s_r)$, with $s_b\in B$ and $s_r\in R$, such that $s_b$ and $s_r$ intersect twice. It suffices to compute the number of arc pairs of the two types.
Counting type (1.1.1) arc pairs can be done in $O(n\log n)$ time by an algorithm of
Agarwal~\cite{ref:AgarwalPa902} (see Section~6~\cite{ref:AgarwalPa902}; the algorithm is for line segments but also applicable to arcs of type (1.1.1)). In what follows, we focus on counting the arc pairs of type (1.1.2).

Agarwal, Pellegrini, and Sharir~\cite{ref:AgarwalCo93} proved that (see Lemma 4.1~\cite{ref:AgarwalCo93}) two arcs $s$ and $s'$ intersect twice if and only if $s$ intersects the underlying circle $\alpha(s')$ of $s'$ twice and $s'$ intersects $\alpha(s)$ twice. Note that each arc $s$ of $B\cup R$ does not span more than a semicircle of $\alpha(s)$.
Combining with Observation~\ref{obser:60}, we obtain the following observation.

\begin{observation}\label{obser:70}
For an arc $s_b\in B$ and an arc $s_r\in R$, $s_b$ and $s_r$ intersect twice if and only if the following five conditions all hold: (1) the center of $s_r$ is in the wedge $w(s_b)$; (2) the underlying disk of $s_r$ does not contain either endpoint of $s_b$; (3) the center of $s_b$ is in the wedge $w(s_r)$; (4) the underlying disk of $s_b$ does not contain either endpoint of $s_r$; (5) the underlying circles of $s_b$ and $s_r$ intersect.
\end{observation}

Recall that the centers of all red arcs are in the square cell $C_r$.   For each blue arc of $s_b\in B$, define $A(s_b)$ and $A'(s_b)$ in the same way as in the proof of Lemma~\ref{lem:80}. We still call $A'(s_b)$ the interesting region of $s_b$. We define $I$ and $E$ as in Lemma~\ref{lem:80}. Note that $|E|=O(m)$.

If $n^{2}< m\log^{2}m$, then we simply apply Lemma~\ref{lem:80}, which can compute the number of arc pairs of type (1.1.2) in $O(n^2\log^2 n+m\log^4 n)$ time.\footnote{To achieve this runtime, we need to consider $B$ as red arcs and $R$ as blue arcs when applying Lemma~\ref{lem:80}.} Since $n^{2}< m\log^{2}m$, the time is bounded by $O(m\log^4 n)$. In what follows, we assume that $n^{2}\geq m\log^{2}m$.

We compute a hierarchical $(1/r)$-cutting $\Xi_0,\Xi_1,\ldots,\Xi_k$
on the general-arcs of $E$, with $r=(\frac{n^{2}}{m\log^{2}m})^{1/3}$, in $O(mr)$
time by Theorem~\ref{theo:cutting}. For each cell $\sigma\in \Xi_i$, $1\leq i\leq k$, we define in the same way as in Lemma~\ref{lem:80} a canonical pair $(B_{\sigma},R_{\sigma})$, with $B_{\sigma}\subseteq B$ and $R_{\sigma}\subset R$, so that centers of all arcs of $R_{\sigma}$ lie in the interesting region of each arc of $B_{\sigma}$. Specifically, $R_{\sigma}$ consists of the arcs of $R$ whose centers are located in the cell $\sigma$ and $B_{\sigma}$ consists of the arcs of $B$ whose interesting regions contain $\sigma$ but do not contain the parent cell of $\sigma$ in $\Xi_{i-1}$. As in Lemma~\ref{lem:80}, the canonical pairs of all cells of all cuttings $\Xi_i$, $i=1,2,\ldots, k$, can be computed in $O(mr + n\log m)$ time.
For each cell $\sigma$ in the last cutting $\Xi_k$, if $|R_{\sigma}|>n/r^2$, we partition $R_{\sigma}$ into subsets of size at most $n/r^2$ each; we call them the {\em basic subsets} of $R_{\sigma}$.
For each cell $\sigma\in \Xi_k$, we define $B'_{\sigma}$ as the set of arcs of $B$ whose interesting regions have an general-arc crossing $\sigma$ and we explicitly maintain $B'_{\sigma}$. All above can be done in $O(mr + n\log m)$ time.

With the above concepts, to count the number of type (1.1.2) arc pairs, it suffices to solve the following two subproblems: (1) For each canonical subset pair $(B_{\sigma},R_{\sigma})$, $\sigma\in \Xi_i$, $1\leq i\leq k$, compute the number pairs of arcs $(s_b,s_r)$, $s_b\in B_{\sigma}$, $s_r\in R_{\sigma}$, such that $s_b$ and $s_r$ satisfy the last three conditions of Observation~\ref{obser:70}; (2) for each cell $\sigma\in \Xi_k$, for each basic subset $R'_{\sigma}$ of $R_{\sigma}$, compute the number of pairs $(s_b,s_r)$, with $s_b\in B'_{\sigma}$ and $s_r\in R'_{\sigma}$, such that $s_b$ and $s_r$ intersect twice.

To solve each subproblem (2) on $B'_{\sigma}$ and $R'_{\sigma}$, we apply the algorithm of Lemma~\ref{lem:80},
which takes $O(|R'_{\sigma}|^2\log^2 |R'_{\sigma}|+|B'_{\sigma}|\log^4|R'_{\sigma}|)$ time.
Note that $|B'_{\sigma}|=O(m/r)$ and $|R'_{\sigma}|=O(n/r^2)$. Hence, the time is bounded by $O(n^2/r^4\cdot \log^2(n/r^2)+\frac{m}{r}\cdot \log^4(n/r^2))$. Since $\Xi_k$ has $O(r^2)$ cells, solving subproblem (2) for all cells of $\Xi_k$ takes $O(n^2/r^2\cdot \log^2(n/r^2)+mr\log^4(n/r^2))$ time, which is $O(n^{2/3}m^{2/3}\log^{10/3}n)$ as $r=\frac{n^{2/3}}{m^{1/3}\log^{2/3}m}$.

The following lemma presents an algorithm to solve subproblem (1).

\begin{lemma}\label{lem:120}
Subproblem (1) on $B_{\sigma}$ and $R_{\sigma}$ for each cell $\sigma\in \Xi_{i}$ can be solved in $O(n_{\sigma}\log^4m_{\sigma}+m_{\sigma}\log^2 (n_{\sigma}+m_{\sigma})+n_{\sigma}^{2/3}m_{\sigma}^{2/3}\log^{10/3} (n_{\sigma}+m_{\sigma}))$ time, where $n_{\sigma}=|R_{\sigma}|$ and $m_{\sigma}=|B_{\sigma}|$.
\end{lemma}
\begin{proof}
For notational convenience, let $B=B_{\sigma}$, $R=R_{\sigma}$, and $n=|B|$, $m=|R|$ in this proof.
Also, the notation in the proof is independent of other parts of the paper.

Recall that the center of each arc of $B$ is in the square cell $C_b$.
For each arc $s_r\in R$, we define $A(s_r)$ in the same way as $A(s_b)$ defined in the proof of Lemma~\ref{lem:80}. Let $A'(s_r)=A(s_r)\cap C_b$. Define $E$ as the set of general-arcs of $A'(s_r)$ for all arcs $s_r\in R$. Hence, $|E|=O(n)$.

If $m^{2}< n\log^{2}n$, we simply apply Lemma~\ref{lem:80}, which can solve the problem in $O(m^2\log^2 m+n\log^4 m)$ time. Since $m^{2}< n\log^{2}n$, the time is bounded by $O(n\log^4 m)$. In what follows, we assume that $m^{2}\geq n\log^{2}n$.

We compute a hierarchical $(1/r)$-cutting $\Xi_0,\Xi_1,\ldots,\Xi_k$
on the general-arcs of $E$, with $r=(\frac{m^{2}}{n\log^{2}n})^{1/3}$, in $O(nr)$
time by Theorem~\ref{theo:cutting}. For each cell $\sigma\in \Xi_i$, $1\leq i\leq k$, we define in a similar way as before a canonical pair $(B_{\sigma},R_{\sigma})$, with $B_{\sigma}\subseteq B$ and $R_{\sigma}\subset R$, so that centers of all arcs of $B_{\sigma}$ lie in the interesting region of each arc of $R_{\sigma}$. Specifically, $B_{\sigma}$ consists of the arcs whose centers are located in the cell $\sigma$ and $R_{\sigma}$ consists of the arcs of $R$ whose interesting regions contain $\sigma$ but do not contain the parent cell of $\sigma$ in $\Xi_{i-1}$. The canonical pairs of all cells of all cuttings $\Xi_i$, $1\leq i\leq k$, can be computed in $O(nr + m\log n)$ time.
For each cell $\sigma$ in the last cutting $\Xi_k$, if $|B_{\sigma}|>m/r^2$, we partition $B_{\sigma}$ into subsets of size at most $m/r^2$ each; we call them the {\em basic subsets} of $B_{\sigma}$.
For each cell $\sigma\in \Xi_k$, we define $R'_{\sigma}$ as the arcs of $R$ whose interesting regions have an general-arc crossing $\sigma$ and we explicitly maintain $R'_{\sigma}$. All above can be done in $O(nr + m\log r)$ time.


To solve our problem of the lemma, it suffices to solve the following subproblems: (1) For each canonical subset pair $(B_{\sigma},R_{\sigma})$, $\sigma\in \Xi_i$, $1\leq i\leq k$, compute the number of pairs of arcs $(s_b,s_r)$, such that $\alpha(s_b)$ and $\alpha(s_r)$ intersect; (2) for each cell $\sigma\in \Xi_k$, for each basic subset $B'_{\sigma}$ of $B_{\sigma}$, compute the number of pairs $(s_b,s_r)$, with $s_b\in B'_{\sigma}$ and $s_r\in R'_{\sigma}$, such that $s_b$ and $s_r$ intersect twice.

To solve each subproblem (2) on $B'_{\sigma}$ and $R'_{\sigma}$, we apply the algorithm of Lemma~\ref{lem:80}, which takes $O(|B'_{\sigma}|^2\log^2 |B'_{\sigma}|+|R'_{\sigma}|\log^4|B'_{\sigma}|)$ time.
Note that $|R'_{\sigma}|=O(n/r)$ and $|B'_{\sigma}|\leq m/r^2$. Hence, the time is bounded by $O(m^2/r^4\cdot \log^2(m/r^2)+\frac{n}{r}\cdot \log^4(m/r^2))$. Since $\Xi_k$ has $O(r^2)$ cells, solving subproblem (2) for all cells of $\Xi_k$ takes $O(m^2/r^2\cdot \log^2(m/r^2)+nr\log^4(m/r^2))$ time, which is $O(m^{2/3}n^{2/3}\log^{10/3}(n+m))$ as $r=\frac{m^{2/3}}{n^{1/3}\log^{2/3}n}$.

To solve subproblem (1) on $(B_{\sigma},R_{\sigma})$ for each cell $\sigma\in \Xi_i$, $1\leq i\leq k$, we do the following. As in Lemma~\ref{lem:80}, define $\calD$ as the set of disks of radius $2$ centered at the centers of all arcs of $R_{\sigma}$ and $P$ the set of the centers of all arcs of $B_{\sigma}$. Subproblem (1) is equivalent to computing the number of pairs $(p,D)$, $p\in P$, $D\in \calD$, such that $p$ is contained in $D$. To solve the problem, we apply the algorithm of Katz and Sharir~\cite{ref:KatzAn97}, which runs in $O(n_{\sigma}\log n_{\sigma}+m_{\sigma}\log n_{\sigma}+n_{\sigma}^{2/3}m_{\sigma}^{2/3}\log n_{\sigma})$ time, where $m_{\sigma}=|\calD|=|B_{\sigma}|$ and $n_{\sigma}=|P|=|R_{\sigma}|$. Next we analyze the total time for solving all subproblems (1) for all cells $\sigma\in \Xi_i$, $i=1,2,\ldots,k$. We consider the three terms in the time complexity separately.

As discussed in Lemma~\ref{lem:80}, $\sum_{\sigma\in \Xi_i}m_{\sigma}=m$, $n_{\sigma}=O(n/\rho^{i-1})$, and $\Xi_i$ has $O(\rho^{2i})$ cells. Note that $m_{\sigma}\leq m$, $n_{\sigma}\leq n$, and $k=\Theta(\log_{\rho} r)$.
Therefore, $\sum_{1\leq i\leq k}\sum_{\sigma\in \Xi_i}m_{\sigma}\log n_{\sigma}\leq \sum_{1\leq i\leq k} m\log n=O(m\log n\log r)$, and $\sum_{1\leq i\leq k}\sum_{\sigma\in \Xi_i}n_{\sigma}\log n_{\sigma}=O(\sum_{1\leq i\leq k} \rho^{2i}\cdot \frac{n}{\rho^{i-1}}\log n)=O(n\log n\cdot \sum_{1\leq i\leq k}\rho^{i+1})$, which is bounded by $O(nr\log n)$ as $k=\Theta(\log r)$ and $\rho$ is a constant. Similarly, $\sum_{1\leq i\leq k}\sum_{\sigma\in \Xi_i} n_{\sigma}^{2/3}m_{\sigma}^{2/3}$ is big-O of $\sum_{1\leq i\leq k}\sum_{\sigma\in \Xi_i} (\frac{n}{\rho^{i-1}})^{2/3}m_{\sigma}^{2/3}=\sum_{1\leq i\leq k}(\frac{n}{\rho^{i-1}})^{2/3} \sum_{\sigma\in \Xi_i} m_{\sigma}^{2/3}$. As $\sum_{\sigma\in \Xi_i}m_{\sigma}=2m$ and $\Xi_i$ has $O(\rho^{2i})$ cells, by H\"older's Inequality, $\sum_{\sigma\in \Xi_i} m_{\sigma}^{2/3}=O(\rho^{2i}\cdot (\frac{m}{\rho^{2i}})^{2/3})$. Hence, $(\frac{n}{\rho^{i-1}})^{2/3} \sum_{\sigma\in \Xi_i} m_{\sigma}^{2/3}=O((\frac{n}{\rho^{i-1}})^{2/3}\cdot \rho^{2i}\cdot (\frac{m}{\rho^{2i}})^{2/3})$, which is bounded by $O(n^{2/3}m^{2/3})$ as $\rho$ is a constant. We thus obtain that $\sum_{1\leq i\leq k}\sum_{\sigma\in \Xi_i} n_{\sigma}^{2/3}m_{\sigma}^{2/3}=O(n^{2/3}m^{2/3}\log r)$. In summary, the total time for solving all subproblems (1) is bounded by $O(n^{2/3}m^{2/3}\log r\log n+nr\log n+m\log n\log r)$, which is
$O(n^{2/3}m^{2/3}\log^{2}(n+m)+m\log n\log m)$ as $r=\frac{m^{2/3}}{n^{1/3}\log^{2/3}n}$.

In summary, the overall time of the entire algorithm is bounded by $O(n\log^4m+m\log^2 (n+m)+n^{2/3}m^{2/3}\log^{10/3} (n+m))$.
\end{proof}


With Lemma~\ref{lem:120}, we next analyze the total time for solving all subproblems (1) for all cells $\sigma\in \Xi_i$, $i=1,2,\ldots, k$. We consider the three terms in the time complexity of Lemma~\ref{lem:120} separately.

Note that $\sum_{\sigma\in \Xi_i}n_{\sigma}=n$, $m_{\sigma}=O(m/\rho^{i-1})$, $\Xi_i$ has
$O(\rho^{2i})$ cells, $m_{\sigma}\leq m\leq n$, $n_{\sigma}\leq n$, and $k=\Theta(\log_{\rho} r)$.
Therefore, $\sum_{1\leq i\leq k}\sum_{\sigma\in \Xi_i}n_{\sigma}\log^4 m_{\sigma}\leq \sum_{1\leq i\leq k} n\log^4 n=O(n\log^4 n\log r)$, and $\sum_{1\leq i\leq k}\sum_{\sigma\in \Xi_i}m_{\sigma}\log^2 (n_{\sigma}+m_{\sigma})=O(\sum_{1\leq i\leq k} \rho^{2i}\cdot \frac{m}{\rho^{i-1}}\log^2 n)=O(m\log^2 n\cdot \sum_{1\leq i\leq k}\rho^{i+1})$, which is bounded by $O(mr\log^2 n)$ as $k=\Theta(\log_{\rho} r)$ and $\rho>1$ is a constant. Also, $\sum_{1\leq i\leq k}\sum_{\sigma\in \Xi_i} n_{\sigma}^{2/3}m_{\sigma}^{2/3}$ is big-O of $\sum_{1\leq i\leq k}\sum_{\sigma\in \Xi_i} (\frac{m}{\rho^{i-1}})^{2/3}n_{\sigma}^{2/3}=\sum_{1\leq i\leq k}(\frac{m}{\rho^{i-1}})^{2/3} \sum_{\sigma\in \Xi_i} n_{\sigma}^{2/3}$. As $\sum_{\sigma\in \Xi_i}n_{\sigma}=n$ and $\Xi_i$ has $O(\rho^{2i})$ cells, by H\"older's Inequality, $\sum_{\sigma\in \Xi_i} n_{\sigma}^{2/3}=O(\rho^{2i}\cdot (\frac{n}{\rho^{2i}})^{2/3})$. Hence, $(\frac{m}{\rho^{i-1}})^{2/3} \sum_{\sigma\in \Xi_i} n_{\sigma}^{2/3}=O((\frac{m}{\rho^{i-1}})^{2/3}\cdot \rho^{2i}\cdot (\frac{n}{\rho^{2i}})^{2/3})$, which is bounded by $O(n^{2/3}m^{2/3})$ as $\rho$ is a constant. We thus obtain $\sum_{1\leq i\leq k}\sum_{\sigma\in \Xi_i} n_{\sigma}^{2/3}m_{\sigma}^{2/3}=O(n^{2/3}m^{2/3}\log r)$. As such, the total time for solving all subproblems (1) is $O(n^{2/3}m^{2/3}\log^{10/3} n\log r+n\log^4 n\log r+mr\log^2 n)$, which is
$O(n^{2/3}m^{2/3}\log^{13/3}n+n\log^5 n)$ as $r=\frac{n^{2/3}}{m^{1/3}\log^{2/3}n}$.

In summary, the total time of the overall algorithm is $O(n^{2/3}m^{2/3}\log^{13/3}n+n\log^5 n+m\log^4 n)$.
This proves Lemma~\ref{lem:110}.


\subsection{Putting it all together}
\label{sec:wrapup}

The above solves the problem of computing the number of intersections in the square cell $C$ between the arcs of $S(C_1)$ and the arcs of $S(C_2)$. The time of the algorithm is $O((n_1+n_2)^{4/3}\log^{16/3} (n_1+n_2))$, where $n_1=|S(C_1)|$ and $n_2=|S(C_2)|$. The algorithm is based on the assumption that $C_1\neq C_2$. If $C_1=C_2$, then the problem becomes computing the intersections in $C$ among all arcs of $S(C_1)$. Our algorithm can be easily adapted to this case. For example, we can duplicate all arcs of $S(C_1)$ and make the duplications as $S(C_2)$, and then apply the algorithm (one subtle issue that in this case an arc of $S(C_1)$ and its duplication in $S(C_2)$ should not be considered intersected).

Enumerating all pairs $(C_1,C_2)$ of $N(C)$ and solving the problem as above can compute the number of all arc intersections in $C$. Processing all cells $C\in \calC$ can finally compute the number of all arc intersections of $S$ in the plane. As such, the total time of the overall algorithm is $O(\sum_{C\in \calC}\sum_{(C_1,C_2)\in N(C)}(n_1+n_2)^{4/3}\log^{16/3} (n_1+n_2))$. Recall that $|S(C')|=|P(C')|$ for any cell $C'\in \calC$.
Because $|N(C)|=O(1)$, we have $\sum_{(C_1,C_2)\in N(C)}(n_1+n_2)^{4/3}\log^{16/3} (n_1+n_2))=O(n_{C}^{4/3}\log^{16/3} n)$, where $n_C=\sum_{C'\in N(C)}|P(C')|$. By Lemma~\ref{lem:60}(4), $\sum_{C\in \calC}n_C=O(n)$. Hence, $\sum_{C\in \calC}n_C^{4/3}=O(n^{4/3})$. As such, the total time of the overall algorithm is bounded by $O(n^{4/3}\log^{16/3} n)$.

\begin{theorem}\label{theo:30}
The number of intersections among a set of $n$ circular arcs of the same radius in the plane can be computed in $O(n^{4/3}\log^{16/3}n)$ time.
\end{theorem}

\subsection{A further improvement}

We further improve the algorithm when the number of intersections of the arcs is relatively small.

\begin{theorem}\label{theo:300}
The number of intersections among a set of $n$ circular arcs of the
same radius in the plane can be computed in
$O(n^{1+\epsilon}+K^{1/3}n^{2/3}(\frac{n^2}{n+K})^{\epsilon}\log^{16/3}n)$ time, for
any $\epsilon>0$, where $K$ is the number of intersections of all arcs.
\end{theorem}
\begin{proof}
Let $S$ denote the set of all arcs.
Let $r=n^2/(n+K)$.
By Theorem~\ref{theo:cutting}, we first compute a $(1/r)$-cutting
$\Xi$ for $S$ of size $O((n^2/(n+K))^{1+\epsilon})$ in
$O(nr^{\epsilon} +Kr/n)$ time, so that each cell
of $\Xi$ is intersected by at most $n/r=(n+K)/n$ arcs of
$S$. Note that $nr^{\epsilon} +Kr/n=O(n^{1+\epsilon})$. Hence,
constructing the cutting takes $O(n^{1+\epsilon})$ time.
For each cell $\sigma$ of $\Xi$, we compute the number of
intersections of arcs inside $\sigma$, in $O(((n+K)/n)^{4/3}\log^{16/3}n)$
time by Theorem~\ref{theo:30}.
The total time of the overall algorithm
is therefore on the order of $n^{1+\epsilon}+
(\frac{n^2}{n+K})^{1+\epsilon}\cdot (\frac{n+K}{n})^{4/3}\cdot
\log^{16/3}n$, which is bounded by
$O(n^{1+\epsilon}+K^{1/3}n^{2/3}(\frac{n^2}{n+K})^{\epsilon}\log^{16/3}n)$.

The above algorithm depends on the unknown value $K$. To overcome of the issue, we guess the value of $K$ by the trick of doubling. We start with $K'=K_0$ for a constant $K_0$, and run the algorithm with $K'$. If the algorithm takes too long, then our guess is too low and we double $K'$. Using this strategy, the algorithm is expected to stop within a constant number of rounds when $K'$ is larger than $K$ for the first time. 
Hence, the total time is asymptotically the same as if we had plugged in the right value of $K$, except that
we call the cutting construction algorithm for $O(\log K)$ (which is
$O(\log n)$) times. Hence, the total time the algorithm spends on
constructing the cutting is still $O(n^{1+\epsilon})$ as the $\log n$
factor is absorbed by $\epsilon$. Therefore, the total time of the
overall algorithm is still
$O(n^{1+\epsilon}+K^{1/3}n^{2/3}(\frac{n^2}{n+K})^{\epsilon}\log^{16/3}n)$.
\end{proof}

\section{Concluding remarks}
\label{sec:con}

In this paper, we present an $O(n^{4/3}\log^{16/3}n)$ time algorithm for computing the number of intersections among a set of $n$ circular arcs of the same radius, which nearly matches the $\Omega(n^{4/3})$ lower bound of Erickson in the partition algorithm model~\cite{ref:EricksonNe96}.
Our algorithm involves a partition of the problem into several cases. While this partition makes the algorithm somewhat tedious, it appears necessary to treat these different cases separately, which may be due to the nature of the problem. Indeed, this partition is one of the main reasons that our algorithm is superior to the previous work of Agarwal, Pellegrini, and Sharir~\cite{ref:AgarwalCo93}. It remains unclear whether this partition can be simplified.

Another more general open question is whether the polylogarithmic factor in the runtime can be further improved. The bottleneck of our algorithm lies in the algorithm of Lemma~\ref{lem:110}. More specifically, the relatively big polylogarithmic factor in the runtime of the algorithm for Lemma~\ref{lem:110} is due to the multiple levels of the algorithm, which is further due to the five conditions of Observation~\ref{obser:70} for having an arc $s_b\in B$ and an arc $s_r\in R$ intersect twice. Therefore, to improve the algorithm, one direction is to find better ways to determine that two arcs $s_b\in B$ and $s_r\in R$ intersect twice. It is not clear whether this is possible. If all five conditions of Observation~\ref{obser:70} have to be used, then it seems difficult to have any further improvement.

Some of our techniques and observations may be useful for solving other related problems. We demonstrate two examples below.

\paragraph{The bichromatic problem.}
It is straightforward to adapt our algorithm to solving the
bichromatic problem, i.e., given a set of blue circular arcs and
another set of red circular arcs of the same radius, with $n$ as the
total number of all arcs, compute the number of intersections between red
arcs and blue arcs. Indeed, our algorithm for solving the problem for
a single cell $C\in \calC$ in Sections~\ref{sec:counttwocells},
\ref{sec:type3}, and \ref{sec:type1} is for the bichromatic problem.
More specifically, we first compute the set $\calC$ of square cells in
a similar way as before. Then, for each cell $C\in \calC$, for each
pair of cells $(C_1,C_2)$ of $N(C)$, we compute the number of
intersections between the blue arcs of $S(C_1)$ and the red arcs of
$S(C_2)$ as well as the number of intersections between the red arcs
of $S(C_1)$ and the blue arcs of $S(C_2)$, by applying our algorithm
in Sections~\ref{sec:counttwocells}, \ref{sec:type3}, and
\ref{sec:type1}. In this way, the bichromatic problem can be solved in
$O(n^{4/3}\log^{16/3}n)$ time. The time can then be further reduced to
$O(n^{1+\epsilon}+K^{1/3}n^{2/3}(\frac{n^2}{n+K})^{\epsilon}\log^{16/3}n)$ time, for
any $\epsilon>0$, where $K$ is the number of intersections of all
arcs.

Alternatively, one can solve the bichromatic problem by using our algorithm for the monochromatic case as a black-box. First, we compute the number of intersections of all arcs, denoted by $K$. Then, we compute the number of intersections among the red arcs, denoted by $K_r$, and also compute the number of intersections among the blue arcs, denoted by $K_b$. Observe that $K-K_r-K_b$ is the number of bichromatic intersections. The total runtime of the algorithm is asymptotically the same as that for the monochromatic case.

\paragraph{The segment case.}
Suppose $S$ is a set of $n$ line segments in the plane. To compute the number of intersections of $S$, using Chan and Zheng's recent $O(n^{4/3})$ time algorithm~\cite{ref:ChanHo23} and the techniques of Theorem~\ref{theo:300}, we can solve the problem in $O(n^{1+\epsilon}+K^{1/3}n^{2/3}(\frac{n^2}{n+K})^{\epsilon})$ time. Indeed, applying  Theorem~\ref{theo:cutting} with $r=n^2/(n+K)$, we first compute a $(1/r)$-cutting
$\Xi$ for $S$ of size $O((n^2/(n+K))^{1+\epsilon})$ in
$O(nr^{\epsilon} +Kr/n)$ time, so that each cell
of $\Xi$ is intersected by at most $n/r=(n+K)/n$ segments of
$S$.
Since $nr^{\epsilon} +Kr/n=O(n^{1+\epsilon})$, constructing the cutting takes $O(n^{1+\epsilon})$ time.
For each cell $\sigma$ of $\Xi$, we compute the number of
intersections of segments inside $\sigma$, in
$O(((n+K)/n)^{4/3})$ time by Chan and Zheng's algorithm~\cite{ref:ChanHo23}.
Hence, the total time is on the order of $n^{1+\epsilon}+
(\frac{n^2}{n+K})^{1+\epsilon}\cdot (\frac{n+K}{n})^{4/3}$, which is bounded by
$O(n^{1+\epsilon}+K^{1/3}n^{2/3}(\frac{n^2}{n+K})^{\epsilon})$.
The above algorithm depends on the unknown value $K$. To overcome of
the issue, we again use the doubling technique as in
Theorem~\ref{theo:300}. As the analysis in
Theorem~\ref{theo:300}, the total time of the algorithm is
still $O(n^{1+\epsilon}+K^{1/3}n^{2/3}(\frac{n^2}{n+K})^{\epsilon})$.

We also remark that by plugging Chan and Zheng's algorithm~\cite{ref:ChanHo23} into a randomized algorithm of de Berg and Schwarzkopf~\cite{ref:deBergCu95} (i.e., replacing Chazelle's algorithm~\cite{ref:ChazelleCu93} with Chan and Zheng's algorithm~\cite{ref:ChanHo23} in Theorem~3~\cite{ref:deBergCu95}), the number of the intersections of $S$ can be computed in $O(n\log n\log K+K^{1/3}n^{2/3})$ expected time. The algorithm is similar as above, except that it uses a cutting with slightly better bounds; specifically, de Berg and Schwarzkopf (Theorem~1~\cite{ref:deBergCu95}) gave a randomized algorithm that can construct a $(1/r)$-cutting of size $O(r+Kr^2/n^2)$ for $S$ in $O(n\log r+Kr/n)$ expected time.

 \bibliographystyle{plainurl}
\bibliography{reference}

%
%

\end{document}